\documentclass[11pt]{article}
\usepackage{graphicx}
\graphicspath{{images/}{../images/}}
\usepackage[T1]{fontenc}
\usepackage[utf8]{inputenc}
\usepackage{amssymb}
\usepackage{amsfonts}
\usepackage{dsfont}
\usepackage{mathtools}
\usepackage{amsthm}
\usepackage{amsmath}
\usepackage{textcomp}
\usepackage{xcolor}
\usepackage{color,soul}
\usepackage{eurosym}
\usepackage{chngcntr}
\usepackage[round]{natbib}
\usepackage[english=british]{csquotes}
\usepackage{enumitem}
\usepackage{subcaption}

\counterwithout{table}{section}
\counterwithout{table}{subsection}

\usepackage{geometry}
\usepackage{blindtext}
\usepackage{hyperref}
\usepackage{cleveref}
\hypersetup{
    colorlinks=true,
    citecolor=blue,
    filecolor=blue,
    linkcolor=blue,
    urlcolor=blue
}

\DeclareMathOperator*{\esssup}{ess\,sup}

\newtheorem{lemme}{Lemma}[section]
\newtheorem{thm}{Theorem}[section]

\newtheorem{cor}{Corollary}[section]
\newtheorem{defi}{Definition}[section]
\newtheorem{rem}{Remark}[section]
\newtheorem{assume}{Assumption}

\newcommand{\hs}{\vspace{3mm}}

\begin{document}

\title{Competition and Incentives in a Shared Order Book}

\author{René \textsc{Aïd}\footnote{Université Paris Dauphine, PSL Research University, LEDa UMR 8007-260, 75775 Paris Cedex 16, rene.aid@dauphine.psl.eu. René Aïd thanks the financial support of the {\em Finance and Sustainable Development EDF-CA CIB Chair}, the {\em Finance for Energy Market} Research Initiative, the French ANR PEPR Math-Vives project MIRTE ANR-23-EXMA-0011.} \qquad Philippe \textsc{Bergault}\footnote{Université Paris Dauphine, CEREMADE, 75775 Paris Cedex 16, bergault@ceremade.dauphine.fr.} \qquad Mathieu \textsc{Rosenbaum}\footnote{Université Paris Dauphine, CEREMADE, 75775 Paris Cedex 16, rosenbaum@ceremade.dauphine.fr.} }
\date{}
\maketitle

\begin{abstract}
Recent regulation on intraday electricity markets has led to the development of shared order books with the intention to foster competition and increase market liquidity. In this paper, we address the question of the efficiency of such regulations by analysing the situation of two exchanges sharing a single limit order book, ie a quote by a market maker can be hit by a trade arriving on the other exchange. We develop a Principal-Agent model where each exchange acts as the Principal of her own market maker acting as her Agent. Exchanges and market makers have all CARA utility functions with potentially different risk-aversion parameters. In terms of mathematical result, we show existence and uniqueness of the resulting Nash equilibrium between exchanges, give the optimal incentive contracts and provide numerical solution to the PDE satisfied by the certainty equivalent of the exchanges. From an economic standpoint, our model demonstrates that incentive provision constitutes a public good. More precisely, it highlights the presence of a competitiveness spillover effect: when one exchange optimally incentivizes its market maker, the competing exchange also reaps indirect benefits. This interdependence gives rise to a free-rider problem. Given that providing incentives entails a cost, the strategic interaction between exchanges may lead to an equilibrium in which neither platform offers incentives—ultimately resulting in diminished overall competition.
\end{abstract}

\vspace{3mm}

\setlength\parindent{0pt}


\vspace{2mm}

\tableofcontents

\section{Introduction}

In the last decade, intraday electricity markets have attracted growing attention from both market operators and academics. This trend has been fueled by the massive deployment of intermittent renewable energy sources across the European Union, which has significantly increased the uncertainty in supply-demand imbalances faced by electricity market participants. Intraday trading offers all market players -- producers, suppliers, and intermediaries alike -- an efficient tool to adjust their positions as updated information about weather conditions and consumption patterns becomes available.

\hs

To improve the efficiency and liquidity of these markets, new regulations have been adopted at the European level. In particular, Regulation (EU) 2019/943 (Article~7) has introduced the possibility for shared order books between trading platforms operating in different countries. Under this mechanism, an order submitted on one exchange is visible and executable by participants of another exchange, regardless of the physical constraints (such as saturated cross-border transmission capacity). This setup aims at fostering integration and competition between exchanges by removing artificial barriers to trading.

\hs

To understand the implications of this regulation, it is helpful to recall how order books function. In continuous double auction markets -- used in most financial and energy markets -- transactions occur through a limit order book, a centralized system that aggregates limit orders, i.e., orders to buy or sell a given quantity at a specified price. Market makers or liquidity providers post such limit orders -- offering to buy (bid) or sell (ask) at specific prices -- thus providing liquidity to the market. Other participants submit market orders, which are executed immediately against the best available limit orders on the opposite side of the book. Matching occurs according to price priority: market orders consume liquidity by executing against the best available price, then the next best, and so on, until the entire quantity is fulfilled or no further matching orders remain. Liquidity is thus endogenously provided, and the efficiency of the market critically depends on the incentives for market makers to participate.

\hs

A shared order book goes one step further: instead of each exchange operating its own independent book, several platforms agree to synchronize their order books so that any order, regardless of where it was posted, can be matched by incoming orders from any of the participating exchanges. This leads to a situation in which quotes posted on one venue can be hit by orders arriving on another. While this structure increases accessibility and reduces fragmentation, it also introduces new strategic interactions between exchanges.

\hs

This setup closely resembles the functioning of the European intraday electricity market, where two major exchanges -- EPEX Spot and Nord Pool -- operate across multiple countries and account for the overwhelming majority of trading volume. As of 2024, companies' website provides the following information on intraday trading volume: 215~TWh for EPEX Spot and 114~TWh for NordPool. Although several other smaller platforms exist, most liquidity and trading activity is concentrated on these two venues. In practice, market participants -- whether they are liquidity takers (such as utilities or large industrial consumers) or liquidity providers (market makers and trading firms) -- tend to have a natural exchange to which they are historically or geographically attached. Typically, utilities trade on the geographical zones where their production assets are located. While switching between exchanges is technically possible, it is often operationally costly. This observation motivates one of the core assumptions of our model: each exchange is associated with a dedicated market maker, who does not switch between platforms. Similarly, liquidity takers are assumed to interact with a specific exchange, although the execution of their orders may occur on either venue due to the shared order book. Figure \ref{fig:slob} illustrates how liquidity provision and consumption operate across multiple exchanges connected via a shared order book. In this example, a market sell order submitted through Exchange $E_1$ is executed on Exchange $E_0$ at the best bid which is $99.5$ posted by Market Maker $M_0$. if the order were not shared, the sell order going only through exchange E$_1$ would have been executed at the bid price of Market Maker M$_1$ at $99.2$.

\begin{figure}[htb]
  \centering
  \includegraphics[width=0.96\textwidth]{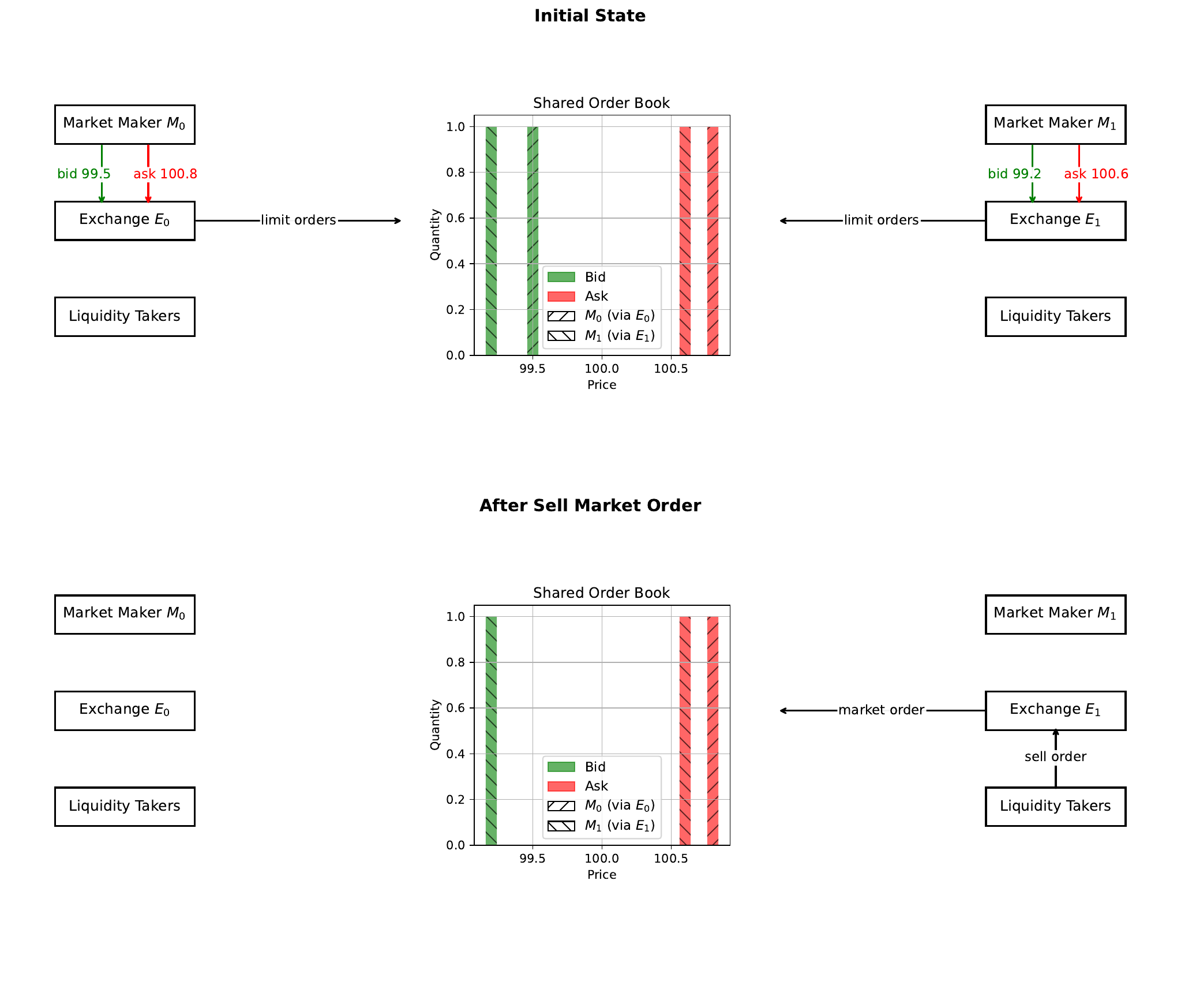}
  \caption{Mechanism of a shared limit order book. The top panel shows initial liquidity provision by two market makers on different exchanges. The bottom panel shows a market sell order submitted through Exchange $E_1$ that executes against the best bid which is quoted by Market Maker M$_0$.}
  \label{fig:slob}
\end{figure}

While our primary focus is on the European intraday electricity market, similar incentive issues arise in other trading environments, most notably in decentralized finance (DeFi). In particular, the use of so-called aggregators -- such as 1inch or Matcha -- effectively connects multiple decentralized exchanges (DEXs), allowing liquidity takers to access the best available quotes across platforms without directly interacting with each one. This introduces a form of competition between venues akin to a shared order book. However, a key distinction lies in the routing and fee structure: in DeFi, liquidity takers interact with the aggregator, which optimally splits the trade across DEXs, but the transaction fees are still collected by the venue where the liquidity is posted. 

\hs

A helpful analogy can be drawn with ride-hailing platforms such as Uber and Lyft. Imagine a situation where these two platforms agree to share their driver availability: a user opening the Uber app could also see and book drivers registered on Lyft, and vice versa. The passenger would then choose the driver offering the best price or shortest expected arrival time (ETA), regardless of their original platform. However, the fees collected from the ride would still be partially attributed to the platform on which the driver is registered. This mirrors the functioning of a shared order book: liquidity takers (passengers) gain access to the best available liquidity (drivers) across platforms, while liquidity providers (market makers or drivers) remain tied to a specific venue, and the economic benefit (fees) depends both on the platform through which the trade is routed and where the liquidity originates.

\hs 

At first glance, one might expect that removing barriers to trading -- as the shared order book mechanism does -- should unambiguously improve market efficiency and social welfare. However, a key insight is that such integration may alter the incentives of exchanges to actively promote liquidity provision. If an exchange anticipates that its members will benefit from the liquidity provided by a competing venue, it may reduce the incentives offered to its own market maker, effectively relying on the competitor’s efforts. This potential for free-riding raises an important question: does the sharing of order books between competing exchanges increase overall liquidity and welfare, or can it lead to negative externalities, such as reduced incentives, lower liquidity, and increased social cost?

\hs

We address this problem by developing a continuous-time model of strategic interaction between two exchanges over a finite trading horizon for a single asset. The fundamental price of the asset is modeled as an arithmetic Brownian motion. Each exchange relies on a dedicated market maker to provide liquidity to its members. Market makers control liquidity by quoting bid and ask prices around the fundamental price: the tighter the spread, the higher the rate of order arrivals. A market maker’s profit and loss consists of three components: gains from the bid-ask spread, inventory costs due to price fluctuations, and incentive payments received from the exchange. The exchange’s profit and loss includes the transaction revenue (assumed linear in volume) and the cost of incentivizing liquidity provision. Incentives take the form of payments indexed on the arrival of orders on \emph{both} exchanges and on the fundamental price of the asset. All players -- exchanges and market makers -- are assumed to have exponential (CARA) utility functions.

\hs

{\bf Results}  First, in terms of mathematical results, we show existence and uniqueness of the resulting Nash equilibrium between exchanges, give the optimal incentive contracts and provide numerical solution to the PDE satisfied by the certainty equivalent of the exchanges. Second, we provide clear evidence that incentive provision to market making in the context of competing exchange is a public good. The costly incentive mechanism implemented by a single exchange largely benefit to a passive competitor, ie a competitor who does not implement her own incentive mechanism. This result is a simple application of a {\em spill-over competitive effect}.  Indeed, incentives provided by one exchange to her market-maker induces more aggressive quotes, which leads to an increase in market orders, which in turns, increases the competitive pressure to the other market maker who is compelled to adjust its own quotes to maintain its level of market orders. The presence of this competitiveness spill-over effect makes incentive provision a public good. Both platforms would personally benefit from the incentives provided by the other exchange to her market maker. But, because the provision of incentive is costly (the platform has to pay the reservation utility of the market maker), an exchange may benefit from not incentivizing her market maker and letting the other platform pay for incentive provision. It is a free-rider problem. As often in the provision of public good, the consequence may be dramatic at equilibrium, as the Nash equilibrium results in no incentive provision by any exchange.

\hs

{\bf Related literature} Market design and liquidity provision have been extensively studied in the literature. \cite{Glosten94} studies the adoption of the electronic order book and shows that it yields many benefits in terms of liquidity compared to other types of market organization. \cite{Foucault16} studies the impact of news on liquidity and show that stocks with more informative news are more liquid. Our study focuses on competition between exchanges, a topic with many open questions presented in  \cite{Cantillon11}. \cite{Colliard12} shows that, although competition increases liquidity and reduces trading fees, it is not always beneficial to investors who are driven to post limit orders with a smaller execution probability. Competition between dealers in OTC markets is studied in \cite{Lester15}. An equilibrium model for coexisting exchanges is presented in \cite{Malamud17}, but it does not account for the possibility of a shared order book.\\ 

This notion of shared order book is particularly relevant in the case of European electricity markets, as mentioned above. The behavior of those markets has recently been an active field of research. \cite{Favetto19} and \cite{Graf21} utilize Hawkes processes to show that, in the intraday market, the order book activity increases exponentially as time approaches maturity. A model for the dynamics of intraday electricity prices has been proposed in \cite{Deschatre24} using a common shock Poisson model that allows the authors to reproduce the correlation structure and the increase of volatility as time to maturity decreases, also known as the Samuelson effect. An alternative model based on marked Hawkes processes was proposed in \cite{Deschatre22} and also allows the authors to reproduce many stylized facts such as the Samuelson effect. Finally, \cite{Bergault24} introduced a model for sparse order books particularly suited for intraday electricity markets.\\

Our work is closely related to that of \cite{ElEuch21}, that combines the \cite{Avellaneda08} model and a Principal-Agent framework with drift control,\footnote{See \cite{Sannikov08} and \cite{Elie19contracting,Elie19tale} for the theoretical foundations of Principal-Agent frameworks.} to increase market liquidity on a single asset by proposing a contract to a market maker. In their formulation, the exchange offers a contract, resulting from a Stackelberg equilibrium between the market maker and the exchange, to the market maker in order to decrease the bid-ask spread offered on the asset and therefore increase liquidity. The optimal contract proposed is a remuneration indexed on the number of transactions realized on the asset as well as its dynamics. This contract is obtained in semi closed-form, depending on the market maker’s inventory trajectory and market volatility. Numerical results show that the contract effectively reduces the spreads and subsequently increases liquidity provision. This work has been extended in various ways, taking into account several market participants, the specificities of options market and enabling trading on dark liquidity pools, see \cite{Baldacci21,Baldacci23market,Baldacci24how} (see also \cite{Aqsha25} for a recent application to decentralized markets). A different approach to make-take fees based on SPDE control theory is proposed in \cite{Baldacci24}. An other interesting study on make-take fees is that of \cite{Foucault13}. Finally, an other tool for exchanges to increase liquidity on a given market is that of setting an optimal tick size: several works have been published in that direction, see for instance \cite{Dayri15} and  \cite{Baldacci23}.\\

\section{The model}

We consider two exchanges or platforms, hereafter respectively denoted by $E_0$ and $E_1$, sharing a common limit order book. Each exchange has one designated market maker, respectively denoted by $M_0$ and $M_1$. Sharing a common limit order book means that a quote from the market maker $M_0$ can be hit by a trade going through his exchange~$E_0$ or from the other exchange $E_1$. In our model, we introduce a measure of the connection between the two exchanges that allows to assess the effect of connection of liquidity and exchanges value functions. In the absence of connection, a market maker can benefit from quotes less competitive that the other market maker. When exchanges are perfectly connected, a market maker is not able to receive more than the best of the two quotes.  We turn now to the precise description of the model. The section~\ref{ssec:setting} provides the description of the dynamics of the inventory of the two market makers and provides the probabiliist framework of the model. The section~\ref{ssec:no-incentive} presents a first benchmark situation when the exchanges do not provide any incentives to their market makers. Sections~\ref{ssec:1-incentive} and \ref{ssec:2-incentive} presents the situation with respectively only one exchange providing optimal incentive to her market maket and the two exchange engaged in a fierce competition where each provides an optimal incentive to her market maker.

\subsection{General setting}\label{ssec:setting}

Let $T>0$ be a final time, $\Omega_c$ the set of continuous functions from $[0,T]$ into $\mathbb R$, $\Omega_d$ the set of piecewise constant càdlàg functions from $[0,T]$ into $\mathbb N$, and $\Omega = \Omega_c \times \big( \Omega_d \big)^8$ with corresponding Borel algebra $\mathcal F$. The observable state is the canonical process $ (\chi_t)_{t \in [0,T]} = \big(W_t, N^{b,0,0}_t, N^{b,0,1}_t,   N^{b,1,0}_t, N^{b,1,1}_t, N^{a,0,0}_t, N^{a,0,1}_t,   N^{a,1,0}_t,  N^{a,1,1}_t \big)_{t \in [0,T]}$ of the measurable space $(\Omega, \mathcal F)$, with
\[
 W_t(\omega) := w(t), \; N^{b,i,j}_t(\omega) :=  n^{b,i,j}(t), \; N^{a,i,j}_t(\omega) :=  n^{a,i,j}(t), \; \text{for all } (i,j) \in \{0,1\}^2, \text{for all } t \in [0,T],
\]
with $\omega := (w, n^{b,0,0}, n^{b,0,1}, n^{b,1,0}, n^{b,1,1}, n^{a,0,0}, n^{a,0,1}, 
n^{a,1,0}, n^{a,1,1}) \in \Omega.$\\

Let $\bar q>0$.We introduce the unique probability measure $\mathbb P$ on $(\Omega, \mathcal F )$ such that, under $\mathbb P$, $W$ is a standard Brownian motion, for $(i,j) \in \{0,1 \}^2$, $N^{b,i,j}$ and $N^{a,i,j}$ are Poisson processes with respective intensities $\mathds{1}_{\{  N^{b,i,0}_{t-} + N^{b,i,1}_{t-} - N^{a,i,0}_{t-} - N^{a,i,1}_{t-}  < \bar q \}}$ and $\mathds{1}_{\{  N^{b,i,0}_{t-} + N^{b,i,1}_{t-} - N^{a,i,0}_{t-} - N^{a,i,1}_{t-}  > - \bar q \}}$, and $W,$ $N^{b,i,j}$, and $N^{a,i,j}$ are independent. We endow the space $(\Omega, \mathcal F )$ with the $\mathbb P-$augmented canonical filtration $\mathbb F := \left(\mathcal F_t  \right)_{t\in [0,T]} $ generated by $ (\chi_t)_{t \in [0,T]}$.\\

The trading activity consists of a single risky asset whose price $S$ is given by
\[
S_t := S_0 + \sigma W_t, \quad t\in [0,T],
\]
with initial price $S_0>0$ and volatility $\sigma>0$ given. The two market makers propose bid and ask quotes for the asset on a continuous basis. More precisely, the bid and ask quotes of market maker $i\in \{0,1\}$ are respectively given by
\[
P^{b,i}_t := S_t - \delta^{b,i}_t \quad \text{and} \quad P^{a,i}_t := S_t + \delta^{a,i}_t \quad t \in [0,T].
\]
We define the following aggregated counting processes:
\[
N^{b,i} := N^{b,i,0} + N^{b,i,1}, \quad N^{a,i} := N^{a,i,0} + N^{a,i,1}, \quad i\in \{0,1\},
\]
which represents respectively the number of trades at the bid and at the ask for the market maker $i$ resulting from the sum of market orders going through exchanges $0$ and $1$.\footnote{For instance, the situation described in Figure \ref{fig:slob} corresponds to a jump of $N^{b,0,1}$.} Besides, we define the inventory processes of the market makers, respectively denoted by $(Q^{0}_t)_{t \in [0,T]}$ and $(Q^{1}_t)_{t \in [0,T]}$, given by
\[
Q^{i}_t := N^{b,i}_t - N^{a,i}_t, \quad i\in \{0,1\}, \ t \in [0,T].
\]
We define the set of admissible controls for the market makers as\footnote{The last constraint is a mathematical device: it ensures that when a market maker reaches their risk limit and is therefore unable to trade on one side of the book, the corresponding (fictitious) quote is set to its maximum value. This prevents the inactive market maker from interfering with the trading activity of the other market maker who remains active.}
\begin{align*}
\mathcal A := \Bigg\{ \delta = \left(\delta^{b,i}_t, \delta^{a,i}_t \right)_{t \in [0,T]}^{i \in \{0,1\}} : \mathbb R^4-&\text{valued predictable processes bounded by } \delta_{\infty},\\
&\qquad \qquad \mathbb P \left( \delta^{k,i}_t = \delta_{\infty}  \left| \phi(k) Q^i_t = \bar q \right. \right) = 1 \text{ for } k \in \{b,a\}\Bigg\},
\end{align*}
where $\delta_{\infty}>0$ is a given constant, $\phi(b) = 1,$ $\phi(a) = -1$.\\

Finally, for $j\in \{0,1\}$, we define the map $\Lambda^{j}: \mathbb R  \rightarrow \mathbb R_+$ such that:
\begin{align}\label{intens_func}
&\Lambda^{j}(d) := A^j \exp \Big(-\frac{\kappa}{\sigma} \big(d + c^j \big) \Big),
\end{align}
where $A^j, \kappa,$ and $c^j$ are positive constants. 

\hs

We introduce for each $\delta \in \mathcal A$ the following Doléans-Dade exponential process $(L^\delta_t)_{t\in [0,T]}$ by:
\[
L^\delta_t := \exp \Bigg( \sum_{ \substack{(i,j) \in \{0,1\} \\ k\in \{b,a\}} }  \bigg( \int_0^t \log \Big( \Lambda^{j}\big(\delta^{k,i}_s\big) \Big) \mathrm dN^{k,i,j}_s - \int_0^t \Big( \Lambda^{j}\big(\delta^{k,i}_s\big)-1\Big) \mathrm ds  \bigg)  \Bigg),
\]
and define the probability measure $\mathbb P^\delta$ given by $ \frac{\mathrm d\mathbb P^\delta}{\mathrm d\mathbb P} := L^\delta_T.$ In particular, all the probability measures $\mathbb P^\delta$ are equivalent, and we therefore use the notation \textit{a.s.} for almost surely without ambiguity.\\

To sum up, under $\mathbb P^{\delta},$ for all $(i,j) \in \{0,1\}^2$, the processes $N^{b,i,j}$ and $N^{a,i,j}$ modelling respectively the number of trades at the bid and at the ask for the market maker $i$ resulting from a market order through exchange $j$ have respective intensities $\left(\lambda^{b,i,j,\delta}_t \right)_{t \in [0,T]}$ and $\left(\lambda^{a,i,j,\delta}_t \right)_{t \in [0,T]}$ given by
\begin{align*}
&\lambda^{b,i,j,\delta}_t := \mathds{1}_{\{  Q^{i}_{t-}  < \bar q \}} \Lambda^{j}\big(\delta^{b,i}_t\big), \quad
    \lambda^{a,i,j,\delta}_t := \mathds{1}_{\{Q^{i}_{t-}  >- \bar q \}} \Lambda^{j}\big(\delta^{a,i}_t\big),
\end{align*}
and $W$ is a standard Brownian motion. This type of exponential intensities for order arrivals is standard in the literature (see \cite{Avellaneda08}). The two exchanges differ only by a constant $A^j e^{-\frac{\kappa c^j}{\sigma}}$ on their relative effect on the intensity of orders arrival because the parameter $\kappa$ is the same for both. The parameter $\bar q$ serves as a risk constraint for market makers: it caps their inventory by preventing further buying when their position is too long, and further selling when it is too short.

\hs

We describe the situation where exchange $E_0$ proposes a contract to its market maker $M_0$, and exchange $E_1$ proposes a contract to its market maker $M_1$. As usual in game theory, we use the notation $i,j$ for the two-player game where $j$ means not $i$.


\hs

Consider the exchange $i$. Let us fix $\beta \in (0,1) $. We introduce the cash process $\left(X^{i,\delta}_t \right)_{t\in [0,T]}$ of market maker $M_i$, given by
\begin{align*}
X^{i,\delta}_t &:= \int_0^t \Big(P^{a,i}_s - \beta \big(\delta^{a,i}_s - \underline{\delta^{a}_s}  \big) \Big) \mathrm dN^{a,i}_s - \int_0^t \Big(P^{b,i}_s + \beta \big( \delta^{b,i}_s - \underline{\delta^{b}_s} \big) \Big) \mathrm dN^{b,i}_s\\ 
& = \int_0^t \Big(\delta^{a,i}_s - \beta \big(\delta^{a,i}_s - \underline{\delta^{a}_s}  \big) \Big) \mathrm dN^{a,i}_s + \int_0^t \Big(\delta^{b,i}_s - \beta \big( \delta^{b,i}_s - \underline{\delta^{b}_s} \big) \Big) \mathrm dN^{b,i}_s - \int_0^t S_s\mathrm dQ^{i}_s    
\end{align*}
where for all $s \in [0,T],$ $k\in \{b,a\}$, $\delta^k_s = \left( \delta^{k,i}_s, \delta^{k,j}_s \right),$ and for $d = \left(d^i, d^j \right) \in \mathbb R^2$, we define $\underline d = \min \left(d^i, d^j\right).$

\hs

The parameter $\beta$ introduces a penalty on the remuneration of the market maker when her quote moves away from the best quote. This represents the fact that, when the quote of the market maker is hit despite the fact that it is not at the first limit, it means that an order has been split and therefore she should not get the remuneration associated with a “full” order (recall that in our model, we make the assumption that all orders are of size 1). In a way, the parameter $\beta$ can also be interpreted as a measure of the efficiency of connection between the two exchanges. When $\beta = 0$, the market maker $M_0$ can take full benefit of a higher quote compared to market maker $M_1$, whereas when $\beta = 1$, she never gets more than the smaller quote. In the latter case, the connection between the two exchanges provides the most efficient service to the users.

\hs

From a modeling perspective, the introduction of $\beta$ also plays a crucial technical role. It allows us to account for the partial execution of market orders -- when they are split across multiple quotes -- without explicitly modeling heterogeneous order sizes or implementing a full matching engine. Instead of varying the size of the executed trade, which would significantly complicate the analysis, we penalize the executed price when the quote is not at the top of the book. Although this is a very simplified representation of reality (for instance, in our model, it may happen that the market maker with the largest quote captures a trade at the expense of a more competitive market maker, something that never occurs in an actual limit order book), this simplification makes the model tractable while preserving the key economic intuition behind shared liquidity access. In particular, in a real order book, either the quotes will be identical -- in which case this issue does not really arise -- or they will be separated by at least one tick, which, due to the exponential form of the intensities, substantially reduces the probability of execution for the less competitive market maker.

\hs

Let us define, for a given $\delta \in \mathcal A$, the profit and loss process $\left( PL^{i,\delta}_t \right)_{t\in [0,T]}$ of market maker $M_i$ as 
\[
PL^{i,\delta}_t := X^{i,\delta}_t + Q^{i}_t S_t = \int_0^t \Big(\delta^{a,i}_s - \beta \big(\delta^{a,i}_s - \underline{\delta^{a}_s}  \big) \Big) \mathrm dN^{a,i}_s + \int_0^t \Big(\delta^{b,i}_s - \beta \big( \delta^{b,i}_s - \underline{\delta^{b}_s} \big) \Big) \mathrm dN^{b,i}_s + \int_0^t Q^{i}_s \mathrm dS_s.
\]

In order to increase market liquidity in the limit order book, exchange $E_i$ proposes to market maker $M_i$ an incentive contract given by an $\mathcal F_T-$measurable random variable $\xi^i$. Given the spread vector $\delta^j$ quoted by market maker $M_j$, and the contract $\xi^i$, we define the objective function $J_{M_i}$ of market maker $M_i$ by
\begin{align*}
J_{M_i}\left(\delta^i, \delta^j, \xi^i \right) &:= \mathbb E^{\delta} \Big[ -e^{-\gamma^i \left(PL^{i,\delta}_T + \xi^i \right)} \Big],
\end{align*}
where $\gamma^i>0$ is the absolute risk aversion of the CARA market maker, $\delta = \left(\delta^i, \delta^j\right)$, and $\mathbb E^\delta$ denotes the expectation under probability $\mathbb P^\delta$. The market maker wants to solve the utility maximization problem
\begin{align}\label{M0PB}
    V_{M_i}\left( \delta^j, \xi^i \right) := \underset{\delta^i \in \mathcal A^i \left(\delta^{j}\right)}{\sup} J_{M_i}\left(\delta^i, \delta^j, \xi^i \right),
\end{align}
where $\mathcal A^i \left(\delta^{j}\right) := \left\{ \delta^i : \delta^i \otimes_i \delta^{j} \in \mathcal A \right\},$ $i\in\{0,1\}.$

\hs

We now give a definition of a Nash equilibrium between market makers $M_0$ and $M_1$ in our framework.

\begin{defi}[{\bf Market makers Nash equilibrium}]
\begin{itemize}
    \item[(i)]For a given contract $\xi^0$ proposed by exchange $E_0$, and a given contract $\xi^1$ proposed by exchange $E_1$, and denoting by $\xi := \left(\xi^0, \xi^1 \right)$, a Nash equilibrium between the market makers is a set of controls $\bar \delta(\xi) \in \mathcal A$ such that
\begin{align*}
V_{M_0}\left( \bar \delta^1(\xi), \xi^0 \right) = J_{M_0}\left(\bar \delta^0(\xi), \bar \delta^1(\xi), \xi^0 \right), \quad
V_{M_1}\left( \bar \delta^0(\xi), \xi^1 \right) = J_{M_1}\left(\bar \delta^0(\xi), \bar \delta^1(\xi), \xi^1\right). \end{align*}
    For a given couple of contracts $\xi = \left(\xi^0, \xi^1 \right)$ proposed by the exchanges, we introduce the set $\text{NE}(\xi) \subset \mathcal A$ of the associated Nash equilibria.
    \item[(ii)] Moreover, we say that a Nash equilibrium $\bar \delta(\xi) \in \mathcal A$ is Pareto-optimal if
\begin{align*}
V_{M_0}\left( \bar \delta^1(\xi), \xi^0 \right) = \underset{\delta \in \text{NE}(\xi)}{\sup} J_{M_0}\left(\delta^0, \delta^1, \xi^0 \right), \quad
V_{M_1}\left( \bar \delta^0(\xi), \xi^1 \right) = \underset{\delta \in \text{NE}(\xi)}{\sup} J_{M_1}\left(\delta^0, \delta^1, \xi^1\right).
\end{align*}

For a given couple of contracts $\xi = \left(\xi^0, \xi^1 \right)$ proposed by the exchanges, we introduce the set $\overline{\text{NE}}(\xi)\subset \mathcal A$ of the associated Pareto-optimal Nash equilibria.
\end{itemize}
\end{defi}

Let $i \in \{0,1\}$. Let us define the reservation level $R_i \in \mathbb R$ of market maker $M_i$ as its maximal utility level without contracts, i.e. 
$$R_i := \underset{\bar \delta \in \text{NE}\left((0,0)\right)}{\sup}  J_{M_i}\left( \bar \delta^0, \bar \delta^1, 0  \right),$$
and set $\eta^i>0$ the exchange's absolute risk aversion parameter.\\

We can now define the set $\mathcal C$ of admissible contracts $\xi = (\xi^0, \xi^1)$:
\begin{defi}[{\bf Admissible contracts}]
The set of admissible contracts $\mathcal C$ is defined as the set of $\mathbb R^2-$valued, $\mathcal F_T- $measurable random variables $\xi= (\xi^0, \xi^1)$, such that $\overline{\text{NE}}(\xi) \neq \emptyset$, and for $i\in \{0,1\}$, there exists $\bar \gamma^i > \gamma^i$ and $\bar \eta^i > \eta^i$, $ \forall \delta \in \mathcal A$
\begin{align} \label{integcond0}
\mathbb E^{\delta} \left[ e^{-\bar \gamma^i \xi^i}\right] < +\infty \quad \text{and} \quad \mathbb E^{\delta} \left[ e^{\bar \eta^i \xi^i}\right]< +\infty,    
\end{align}
and 
$$V_{M_i}\left( \bar \delta^i(\xi), \xi^i \right) \ge R_i \quad \forall \bar \delta(\xi) \in \overline{\text{NE}}(\xi).$$
\end{defi}

\hs

Let us first solve the problem of the two market makers when exchanges $E_0$ and $E_1$ propose an arbitrary couple of admissible contract $(\xi^0, \xi^1) \in \mathcal C$ to their respective market makers $M_0$ and $M_1$.\\

We introduce the notations $\mathcal R := \mathbb R^4 \times \mathbb R^4 \times \mathbb R$, $\bar{\mathcal Q}:= [-\bar q, \bar q ]$, and $\mathcal B_{\infty} := [-\delta_{\infty}, \delta_{\infty}]$. We now define the Hamiltonian functions of the market makers.\\

For $(d^0, d^1, z^l, q) \in \mathcal B_{\infty}^2 \times \mathcal B_{\infty}^2 \times \mathcal R \times \bar{\mathcal Q}^2$, we denote $z^l := \left( \left( z^{l,b,i,j} \right)_{i,j = 0,1}, \left( z^{l,a,i,j} \right)_{i,j = 0,1}, z^{l,S}  \right)$ and $d^i := \left(d^{b,i}, d^{a,i} \right)$ $ \forall i,l \in \{0,1\}$. For $q^0 \in \bar{\mathcal Q},$ we define
$$
\mathcal B_{\infty}^2(q^0) = \begin{cases}
\mathcal B_{\infty}^2 \qquad \text{if } q^0 \in (-\bar q, \bar q)\\
\{\delta_\infty\} \times \mathcal B_{\infty} \qquad  \text{if } q^0 = \bar q\\
\mathcal B_{\infty} \times \{\delta_\infty\} \qquad  \text{if } q^0 = -\bar q.
\end{cases}
$$
Finally, we introduce the set $\Psi$ given by
\begin{align*}
    \Psi = \Bigg\{(d^0, d^1, z^0, q) \in \mathcal B_{\infty}^2 \times \mathcal B_{\infty}^2 \times \mathcal R \times \bar{\mathcal Q}^2 \bigg| d^0 \in  \mathcal B_{\infty}^2(q^0), d^1 \in  \mathcal B_{\infty}^2(q^1) \Bigg\}.
\end{align*}

The variable $z^{\ell,k,i,j}$ can be interpreted as a payment rate for market maker $\ell$ for incentives for bid ($k=b$) or ask ($k=a$) depending on what is the observed rate of order arrival $N^{k,i,j}$ to market maker $i$ on exchange $j$.

\begin{defi}[{\bf Market makers Hamiltonian}]
For $( z^0, q) \in \mathcal R \times \bar{\mathcal Q}^2$, and $d^1 \in \mathcal B_{\infty}^2(q^1)$, the Hamiltonian function $H^0$ of market maker $M_0$ is given by
\begin{align}
    H^0 (d^1, z^0, q) := \underset{d^0 \in \mathcal B_{\infty}^2(q^0)}{\sup} h^0(d^0, d^1, z^0, q),
\end{align}
where $h^0$ is defined on $\Psi$ by
\begin{align*}
    h^0(d^0, d^1, z^0, q) := \sum_{k\in \{b,a\}} \Bigg(&  \sum_{j=0}^1 \frac{1 - e^{-\gamma^0 \Big(z^{0,k,0,j} + {d^{k,0}} - \beta \big( d^{k,0} - \underline{d^k} \big) \Big)}}{\gamma^0}  \mathds{1}_{\{q^0 \phi(k)  < \bar q \}} \Lambda^{j}\big(d^{k,0}\big)\\
    & + \sum_{j=0}^1 \frac{1 - e^{-\gamma^0 z^{0,k,1,j} }}{\gamma^0}  \mathds{1}_{\{q^1 \phi(k)  < \bar q \}} \Lambda^{j}\big(d^{k,1}\big) \Bigg).
\end{align*}

The Hamiltonian function $H^1$ of market maker $M_1$ deduced by symmetry.
\end{defi}

For $q \in \bar{\mathcal Q}$, we define 
$$ \mathcal M_2 (\mathcal B_{\infty},q ) = \mathcal B_{\infty}^2(q^0) \times \mathcal B_{\infty}^2(q^1).$$
In order to find a Nash equilibrium between the market makers, we now give a definition of a fixed point of the Hamiltonian vector $\left(H^0(., z^0, q), H^1(.,z^1,q) \right).$

\begin{defi}[{\bf Market makers Hamiltonians fixed points}]
\begin{itemize}
    \item[(i)] For every $(z,q) \in \mathcal R^2 \times \bar{\mathcal Q}^2$, where we denote $z := (z^0, z^1)$, a fixed point of the Hamiltonian functions is given by a matrix $\delta^\star (z,q) \in \mathcal M_2 (\mathcal B_{\infty}, q)$ such that
\begin{align*}    
\delta^{\star 0} (z,q) \in \underset{\delta^0 \in \mathcal B_{\infty}^2(q^0)}{\mathrm{argmax}}\ h^0(\delta^0, \delta^{\star 1}(z,q), z^0, q), \quad
\delta^{\star 1} (z,q) \in \underset{\delta^1 \in \mathcal B_{\infty}^2(q^1)}{\mathrm{argmax}}\ h^1(\delta^{\star 0}(z,q), \delta^{ 1}, z^1, q).
\end{align*}

We denote by $\mathcal E (z,q)$ the set of all such fixed points. The corresponding set of maps is denoted by $\mathcal E.$
 \item[(ii)] We also introduce the two functions $\mathcal H ^0, \mathcal H^1 : \mathcal R^2 \times \bar{\mathcal Q}^2 \rightarrow \mathbb R$ given by
 \begin{align*}
 \mathcal H ^0(z,q) = \underset{\delta^\star(z,q) \in \mathcal E(z,q)}{\sup} H^0 \left(\delta^{\star 1}(z,q), z^0,q \right),
 \quad 
 \mathcal H ^1(z,q) = \underset{\delta^\star(z,q) \in \mathcal E(z,q)}{\sup} H^1 \left(\delta^{\star 0}(z,q), z^1,q \right).
 \end{align*}

 We say that $\delta^\star (z,q)\in \mathcal M_2 (\mathcal B_{\infty}, q)$ is an optimal fixed point of the Hamiltonian functions if 
 \begin{align*}
 \delta^{\star} (z,q) \in \underset{\delta^\star(z,q) \in \mathcal E(z,q)}{\mathrm{argmax}}\ H^0 \left(\delta^{\star 1}(z,q), z^0,q \right),
\quad
\delta^{\star} (z,q) \in \underset{\delta^\star(z,q) \in \mathcal E(z,q)}{\mathrm{argmax}}\ H^1 \left(\delta^{\star 0}(z,q), z^1,q \right).
\end{align*}
We denote by $\bar{\mathcal E}(z,q)$ the set of all such optimal fixed points. The corresponding set of maps is denoted by $\bar{\mathcal E}.$
\end{itemize}
\end{defi}

\hs

For $z = (z^0,z^1) \in \mathcal R^2$, with $$z^l =  \left( \left( z^{l,b,i,j} \right)_{i,j = 0,1}, \left( z^{l,a,i,j} \right)_{i,j = 0,1}, z^{l,S}  \right)$$
for $l,i \in \{0,1\}$, we introduce the notation $z^{l,k,i,:}$ for $l \in \{0,1\}$ and $k\in\{b,a\}$ such that
$$z^{l,k,i,:} = \left(z^{l,k,i,0},z^{l,k,i,1}\right),$$
and the notation $\bar z^{k,:}$ for $k\in\{b,a\}$ such that
$$\bar z^{k,:} = \left(z^{0,k,0,:}, z^{1,k,1,:}\right).$$

Besides, we introduce the notations. For $i\in \{0,1\}$ and $k\in\{b,a\}$, we define the functions $\tilde \Gamma^{k,i} : \mathbb R^2 \rightarrow \mathcal B_{\infty}$ and $\tilde \Gamma^{k,i}_\beta : \mathbb R \times \mathbb R^2 \rightarrow \mathcal B_{\infty}$ such that, for every $(z,d)\in \mathcal R^2 \times \mathbb R$, we have
\begin{align*}
    \tilde \Gamma^{k,i}(z^{i,k,i,:}) &:= (-\delta_\infty) \vee \left\lbrace\frac{1}{\gamma^i} \left( \log \bigg( 1 + \frac{\sigma \gamma^i}{\kappa}  \bigg) + \log \Bigg( \frac{\sum_{j=0}^1 A^j e^{-\frac{\kappa}{\sigma}c^j - \gamma^i z^{i,k,i,j}}}{\sum_{j=0}^1 A^j e^{-\frac{\kappa}{\sigma}c^j}}  \Bigg)   \right) \right\rbrace \wedge \delta_\infty,\\
    \tilde \Gamma^{k,i}_\beta (d,z^{i,k,i,:}) &:=(-\delta_\infty) \vee \left\lbrace - \frac{\beta}{1-\beta} d + \frac{1}{\gamma^i (1-\beta)} \left( \log \bigg( 1 + \frac{(1-\beta)\sigma \gamma^i}{\kappa}  \bigg) + \log \Bigg( \frac{\sum_{j=0}^1 A^j e^{-\frac{\kappa}{\sigma}c^j - \gamma^i z^{i,k,i,j}}}{\sum_{j=0}^1 A^j e^{-\frac{k}{\sigma}c^j}}  \Bigg)   \right)\right\rbrace \wedge \delta_\infty.
\end{align*}
We also define the functions $\Gamma^{k,i} : \mathbb R^2 \times [-\bar q, \bar q] \rightarrow \mathcal B_{\infty}$ and $\Gamma^{k,i}_\beta : \mathbb R \times \mathbb R^2\times [-\bar q, \bar q] \rightarrow \mathcal B_{\infty} $ such that, for every $(z,d, q)\in \mathcal R^2 \times \mathbb R \times [-\bar q, \bar q]$, we have
\begin{align*}
    \Gamma^{k,i}(z^{i,k,i,:},q) &:=
    \begin{cases}
    \tilde \Gamma^{k,i}(z^{i,k,i,:}) \quad \text{if} \quad q\phi(k)< \bar q,\\
    \delta_\infty \quad \text{otherwise,}
    \end{cases} \qquad 
    \Gamma^{k,i}_\beta (d,z^{i,k,i,:},q) &:=
    \begin{cases}
    \tilde \Gamma^{k,i}_\beta (d,z^{i,k,i,:}) \quad \text{if} \quad q\phi(k)< \bar q,\\
    \delta_\infty \qquad \text{otherwise.}
    \end{cases}
\end{align*}
Finally, we introduce the function $\Delta:\mathcal R^2 \times [-\bar q, \bar q]^2 \mapsto \mathcal M_2(\mathcal B_{\infty})$ such that for every $(z,q) \in \mathcal R^2 \times [-\bar q, \bar q]$:
\begin{align*}
  \Delta^{k,i}(\bar z^{k,:},q) :=
  \begin{cases}
  \Gamma^{k,i}(z^{i,k,i,:},q^i) \quad \text{if} \quad \Gamma^{k,i}(z^{i,k,i,:},q^i)\le \Gamma^{k,1-i}(z^{1-i,k,1-i,:},q^{1-i}),\\
  \Gamma^{k,i}_\beta \left(\Gamma^{k,1-i}(z^{1-i,k ,1-i,:},q^{1-i}),z^{i,k,i,:},q^i\right)\\ \qquad \quad \text{if} \quad \Gamma^{k,1-i}(z^{1-i,k,1-i,:},q^{1-i})< \Gamma^{k,i}(z^{i,k,i,:},q^i)\land \Gamma^{k,i}_\beta \left(\Gamma^{k,1-i}(z^{1-i,k,1-i,:},q^{1-i}),z^{i,k,i,:},q^i\right),\\
  \Gamma^{k,1-i}(z^{1-i,k,1-i,:},q^{1-i}) \quad \text{otherwise.}
  \end{cases}
\end{align*}

\hs

\begin{lemme}\label{FPMM}
It holds that $\bar{\mathcal E} = {\mathcal E} = \{\Delta\}$. 
\end{lemme}

\hs

\begin{rem}{\rm 
    The expression of the fixed-point is made complicated by the necessity of dealing with the bound $\bar q$ of the inventories of the market makers. But, if we consider only the case where inventories have not reached their bounds, the bid-ask spreads implemented by the market makers for a given vector of payment rates $z^{\ell,:}$ reduces to the following expressions. We focus here on the ask side and see that it is given by
\begin{align*}
    \delta^{\star,i,a}(z,q) = \frac{1}{\gamma^i}\log\Big( 1 + \frac{\sigma \gamma^i}{\kappa}  \Big) 
    + \frac{1}{\gamma^i} \log\Big( \frac{\tilde A_0}{A} e^{-\gamma^i z^{i,a,i,0}} 
              + \frac{\tilde A_1}{A} e^{-\gamma^i z^{i,a,i,1}}
          \Big)
\end{align*}
where $A := \sum_{j=0}^1 A^j e^{-\frac{\kappa}{\sigma}c^j}$ and $\tilde A_j := A^j e^{-\frac{\kappa}{\sigma}c^j}.$ The first part is independent of the incentives of the exchange; it is the fundamental spread performed by the market maker in the absence of incentives ($z^{\ell}=0)$. The second part allows the exchange either to reduce the ask ($z^{\ell}>0$) or to increase it ($z^{\ell}<0$). Imagine exchange $E_0$ is less attractive to agents than exchange $E_1$ (i.e. $A^0 \ll A^1$) and she wants to reduce the ask. Then, exchange $E_0$ has to provide an incentive payment  $z^{0,a,0,1}$ for what is happening on exchange $E_1$ much larger than the payment rate $z^{0,a,0,0}$ for what is happening on her own exchange $E_0$ with her own market maker.}
\end{rem}

\hs

Now, we turn to the representation of admissible incentives.  For two $\mathcal R-$valued predictable processes $Z^0= (Z^{0,b}, Z^{0,a}, Z^{0,S})$ and $Z^1= (Z^{1,b}, Z^{1,a}, Z^{1,S})$, denoting $Z := (Z^0, Z^1)$, and for arbitrary constants $Y_0 = \left(Y^0_0, Y^1_0 \right) \in \mathbb R^2$, we introduce a $\mathbb R^2-$valued process $\left( Y^{Y_0,Z}_t \right)_{t\in [0,T]} := \left(Y^{0,Y^0_0, Z}_t, Y^{1,Y^1_0, Z}_t\right)_{t\in [0,T]}$, such that for $t\in [0,T]$, for $m=0,1$, 
\begin{align}
Y_t^{m,Y^m_0,Z} := Y^m_0 +\int_0^t \Bigg( & \sum_{i,j=0}^1 \bigg(Z^{m,b,i,j}_s \mathrm dN^{b,i,j}_s + Z^{m,a,i,j}_s \mathrm dN^{a,i,j}_s \bigg) + Z^{m,S}_s \mathrm dS_s\\
& + \Big( \frac{\gamma^m \sigma^2}{2} \big(Z^{m,S}_s + Q^{m}_s \big)^2 - \mathcal H^m \big( Z_s, Q_s\big) \Big)ds \Bigg).\nonumber
\end{align}

\begin{defi}[{\bf Contract representation}]
A $\mathcal R^2-$valued predictable process $Z= (Z^0, Z^1)$ belongs to the set $\mathcal Z$ if $\xi^0 = Y^{0,Y^0_0, Z}_T$ and $\xi^1 = Y^{1,Y^1_0, Z}_T$  satisfy the integrability conditions  \eqref{integcond0}, and
$$\mathbb E^\delta \left[ \underset{t \in [0,T]}{\sup} e^{-\gamma' Y^{i,Y^i_0, Z}_t}\right] < + \infty, \quad \text{for some } \gamma' > \gamma^i.$$
\end{defi}

$\mathcal Z$ is not empty as it contains all bounded predictable processes. We define the following set:
\begin{align*}
\Xi := \Bigg\{ Y_T^{Y_0, Z} = &\left(Y^{0,Y^0_0, Z}_T, Y^{1,Y^1_0, Z}_T\right) \bigg| Y_0\in \mathbb R^2, Z \in \mathcal Z,\\
&\text{ and } V_{M_i}\left( \bar \delta^{1-i}(Y_T^{Y_0, Z}), Y_T^{i,Y_0^i,Z} \right) \ge R_i  \forall i \in \{0,1\}, \forall \bar \delta(Y_T^{Y_0, Z}) \in \overline{\text{NE}}(Y_T^{Y_0, Z})\Bigg\}.
\end{align*}

\begin{thm}\label{contractrep}
Any couple of contracts $\xi=\left(\xi^0, \xi^1\right)$ is of the form $\xi^0 = Y_T^{0, Y^0_0, Z}$ and $\xi^1 = Y_T^{1, Y^1_0, Z}$ for some $(Y_0, Z) \in \mathbb R^2 \times \mathcal Z$. In particular, $\mathcal C = \Xi$.
\end{thm}

\begin{cor}\label{NashMM}
For any admissible contracts  $Y^{Y_0,Z}_T := \left(Y^{0,Y^0_0, Z}_T, Y^{1,Y^1_0, Z}_T\right) \in \Xi$ offered by the exchanges, there exists an optimal Nash equilibrium given by $\Big(\Delta (Z_t, Q_t) \Big)_{t \in [0,T]} \in \overline{\text{NE}}\left( Y^{Y_0,Z}_T\right)$, where the function $\Delta$ is defined in  \Cref{FPMM}.
\end{cor}

From now on, we make the following assumption:

\begin{assume}
For any admissible contracts  $Y^{Y_0,Z}_T := \left(Y^{0,Y^0_0, Z}_T, Y^{1,Y^1_0, Z}_T\right) \in \Xi$ offered by the exchanges, the market makers follow the optimal bid-ask policies $\Big(\Delta (Z_t, Q_t) \Big)_{t \in [0,T]}$ given above.
\end{assume}

We now denote $\hat\delta (\xi) = \Big(\Delta (Z_t, Q_t) \Big)_{t \in [0,T]}\in \overline{\text{NE}}\left( Y^{Y_0,Z}_T\right)$. 

\begin{rem}\label{val_func_mm_0}
    Under the previous assumption, it is easy to check that for $i\in \{0,1\}$,
$$V_{M_i}\left( \hat \delta^i(Y^{Y_0,Z}_T), Y^{i,Y^i_0, Z}_T \right) = -e^{-\gamma Y^i_0}.$$
\end{rem}

\subsection{Absence of incentives}\label{ssec:no-incentive}

In order to get a benchmark for the next cases, we study the situation in which the exchanges do not propose a contract to their market makers, i.e. $\xi = (\xi^0, \xi^1) = (0,0)$.\\

To find the corresponding optimal Nash equilibrium of \Cref{NashMM}, one needs to find the representation of the couple of contracts $(0,0)$ as in \Cref{contractrep}. Hence, we need to solve the following coupled BSDEs for $m=0,1$:
\begin{align}\label{00BSDEsystem}
\begin{cases}
dY_s^{m} = \sum_{i,j=0}^1 \bigg(Z^{m,b,i,j}_s \mathrm dN^{b,i,j}_s + Z^{m,a,i,j}_s \mathrm dN^{a,i,j}_s \bigg)& + Z^{m,S}_s \mathrm dS_s\\
& + \Big( \frac{\gamma^m \sigma^2}{2} \big(Z^{m,S}_s + Q^{m}_s \big)^2 - \mathcal H^m \big( Z_s, Q_s\big) \Big)ds,\\
Y_T^0 = Y_T^1 = 0.
\end{cases}    
\end{align}

Note that they are coupled through the unknown process $Z = (Z^0,Z^1) \in \mathcal Z$, with $$Z^l =  \left( \left( Z^{l,b,i,j} \right)_{i,j = 0,1}, \left( Z^{l,a,i,j} \right)_{i,j = 0,1}, Z^{l,S}  \right)$$
for $l \in \{0,1\}$.\\

Consider the following system of PDEs for $m=0,1$:
\begin{align}\label{00PDEsystem}
\begin{cases}
0 = \partial_t \hat w_m(t,q) - \frac{\gamma^m \sigma^2}{2} (q^m)^2 +\mathcal H^m \Bigg( \hat \zeta \bigg(\Big(\hat w_l(t,q) \Big)_{l = 0,1}, \Big( \hat w_l(t,q + \phi(k)e^i) \Big)_{\substack{k = b,a\\ l, i =0,1 }} \bigg) ,q \Bigg),\\
\hat w_0(T,q) = \hat w_1(T,q) = 0,
\end{cases}
\end{align}
$\forall (t,q) \in [0,T) \times \bar{\mathcal Q}^2$, where $e^0 = (1,0)^\intercal$, $e^1 = (0,1)^\intercal$,
\begin{align*}
\hat \zeta &\bigg(\Big(\hat w_l(t,q) \Big)_{l = 0,1}, \Big( \hat w_l(t,q + \phi(k)e^i) \Big)_{\substack{k = b,a\\ l, i =0,1 }} \bigg)\\
&:= \left(\hat \zeta^0 \bigg(\hat w_0(t,q) , \Big( \hat w_0(t,q + \phi(k)e^i) \Big)_{\substack{k = b,a\\  i =0,1 }} \bigg), \hat \zeta^1 \bigg(\hat w_0(t,q) , \Big( \hat w_0(t,q + \phi(k)e^i) \Big)_{\substack{k = b,a\\ i =0,1 }} \bigg) \right) \in \mathcal R^2,    
\end{align*}
with
\begin{align*}
    & \hat \zeta^{0,b,i,j} \bigg(\hat w_0(t,q) ,  \hat w_0(t,q + e^i) \bigg) = \hat w_0(t,q + e^i) - \hat w_0(t,q ),\\
    & \hat \zeta^{0,a,i,j} \bigg(\hat w_0(t,q), \hat w_0(t,q -e^i) \bigg) = \hat w_0(t,q - e^i) - \hat w_0(t,q ),\\
    & \hat \zeta^{0,S} = 0,
\end{align*}
and similarly for Exchange $1$.
\hs

The following Lemma is a straightforward application of Cauchy-Lipschitz Theorem and Itô's formula:

\begin{lemme}\label{solE00}
There exists a unique solution $\left(\hat w_0, \hat w_1 \right)$ to the system of PDEs \eqref{00PDEsystem}. Moreover, there exists a unique solution to the system of BSDEs \eqref{00BSDEsystem}, which is given by 
\begin{align*}
    Z^{l,S}_t = 0, \quad Z^{l,b,i,j}_t = \hat w_l\big(t, Q_{t-} + e^i\big) - \hat w_l\big(t,Q_{t-}\big),\quad Z^{l,a,i,j}_t = \hat w_l\big(t, Q_{t-} - e^i\big) - \hat w_l\big(t,Q_{t-}\big),
\end{align*}
$\forall t \in [0,T]$, $\forall (l,i,j) \in \{0,1\}^3$.\\

In particular, $\forall (l,i) \in \{0,1\}^2$, $\forall k \in \{b,a\}$, $Z^{l,k,i,0} = Z^{l,k,i,1}.$
\end{lemme}

This allows to compute the optimal Nash equilibrium of the market makers given in \Cref{NashMM}. 

\begin{rem}
    Following the preceeding Lemma, when there are no incentive mechanism, a market market implements a spread given by
\begin{align*}
    \delta^{\star,i,a}(z,q) = \frac{1}{\gamma^i}\log\Big( 1 + \frac{\sigma \gamma^i}{\kappa}  \Big)  -   z^{i,a,i},
\end{align*}
where $z^{i,a,i}:= z^{i,a,i, 0} = z^{i,a,i, 1}.$
\end{rem}

\subsection{Incentive for only one market maker}\label{ssec:1-incentive}

Let us now turn to the asymmetric case where exchange $E_0$ proposes a contract to its market maker, whereas exchange $E_1$ stays passive, i.e. $\xi^1 = 0.$  The case of competition between exchanges is presented in \Cref{ssec:2-incentive}.

\hs

Inspired by \cite{Baldacci21}, and in accordance with the results of the previous section, to find the optimal contract of exchange $E_0$, we will now restrict ourselves to a subset of admissible contracts under which we can derive the optimal strategy of the exchange in a semi-explicit form. We therefore introduce the following set:
\begin{align*}
\Xi' := \Bigg\{ Y_T^{Y_0, Z} = &\left(Y^{0,Y^0_0, Z}_T, Y^{1,Y^1_0, Z}_T\right) \bigg| Y_0\in \mathbb R^2, Z \in \mathcal Z, Z^{i,k,i,0} = Z^{i,k,i,1} =: Z^{i,k,i}\ \forall i \in \{0,1\}, \forall k \in \{b,a\},\\
&\qquad \qquad \text{ and } V_{M_i}\left( \bar \delta^{1-i}(Y_T^{Y_0, Z}), Y_T^{i,Y_0^i,Z} \right) \ge R_i \,\,\, \forall i \in \{0,1\}, \forall \bar \delta(Y_T^{Y_0, Z}) \in \overline{\text{NE}}(Y_T^{Y_0, Z})\Bigg\}.
\end{align*}

For each market order, the exchange $E_i$ receives a fixed fee $c^i>0.$ Given that $\xi^1 = 0$, the optimization problem of exchange $E_0$ is given by
\begin{align}\label{pb10}
    V_{E_0}(0) := \!\!\!\!\underset{\xi^0 \in  {\Xi'}^0(0)}{\sup} \mathbb E^{\hat \delta\big((\xi^0, 0)\big) } \left[-e^{-\eta^0 \left(c^0\left( N^{b,0,0}_T + N^{b,1,0}_T + N^{a,0,0}_T + N^{a,1,0}_T \right) - \xi^0 \right)} \right],
\end{align}
where ${\Xi'}^i \left(\xi^{1-i}\right) := \left\{ \xi^i: \xi^i \otimes_i \xi^{1-i} \in {\Xi'} \right\},$ $i\in\{0,1\}.$ In view of Remark \ref{val_func_mm_0}, $E_0$ has to choose an initial value $Y^0_0 = -1/\gamma \log(-R^0)$ for his contract's representation, and this term can be factored out of the equation so that we only optimize over the admissible processes $Z^0$.

\hs

As before, we need to get a representation of the contract $\xi^1 = 0$ under the form
\begin{align*}
0 = \xi^1 = Y_T^{1,Y^1_0,Z} := Y^1_0 +\int_0^T \Bigg( & \sum_{i,j=0}^1 \bigg(Z^{1,b,i,j}_s \mathrm dN^{b,i,j}_s + Z^{1,a,i,j}_s \mathrm dN^{a,i,j}_s \bigg) + Z^{1,S}_s \mathrm dS_s\\
& + \Big( \frac{\gamma^1 \sigma^2}{2} \big(Z^{1,S}_s + Q^{1}_s \big)^2 - \mathcal H^1 \big(Z_s, Q_s\big) \Big)ds \Bigg).
\end{align*}
In other words, given the process $Z^0$ representing the contract of exchange $E_0$, with $$Z^0 =  \left( \left( Z^{0,b,i,j} \right)_{i,j = 0,1}, \left( Z^{0,a,i,j} \right)_{i,j = 0,1}, Z^{0,S}  \right),$$
one needs to solve the BSDE
\begin{align}\label{0BSDE}
\begin{cases}
dY_s = \sum_{i,j=0}^1 \bigg(Z^{1,b,i,j}_s \mathrm dN^{b,i,j}_s + Z^{1,a,i,j}_s \mathrm dN^{a,i,j}_s \bigg) + Z^{1,S}_s \mathrm dS_s + \Big( \frac{\gamma^1 \sigma^2}{2} \big(Z^{1,S}_s + Q^{1}_s \big)^2 - \mathcal H^1 \big(Z_s, Q_s\big) \Big)ds,\\
Y_T = 0,
\end{cases}    
\end{align}
where $Z=(Z^0,Z^1)$, with
$$Z^1 =  \left( \left( Z^{1,b,i,j} \right)_{i,j = 0,1}, \left( Z^{1,a,i,j} \right)_{i,j = 0,1}, Z^{1,S}  \right).$$

\hs

Let us define the function $g^{0,S}: \mathbb R \times \bar{\mathcal Q}  \rightarrow \mathbb R$ such that
$$g^{0,S}\left(z^{0,S}, q^0 \right) = \frac{\eta^0 \sigma^2}{2} \left( \gamma^0 \left(z^{0,S} + q^0\right)^2 + \eta^0 \left(z^{0,S}\right)^2 \right) \quad \forall \left(z^{0,S}, q^0 \right) \in \mathbb R \times \bar{\mathcal Q},$$

and the function $G^{0,S}:\bar{\mathcal Q}  \rightarrow \mathbb R$ given by 
$$G^{0,S}(q^0) = \underset{z^{0,S}\in \mathbb R}{\sup} g^{0,S}\left(z^{0,S}, q^0 \right).$$

Let us also define for $k\in\{b,a\}$ the function $g^{0,k}:\mathbb R \times \mathbb R^2 \times \mathbb R \times \bar{\mathcal Q}^2  \times \mathbb R \times \mathbb R \times \mathbb R\rightarrow \mathbb R$ such that
\begin{align*}
& g^{0,k} \left(z^{0,k,0}, z^{0,k,1,:}, z^{1,k,1},q,  y,y', y'' \right)\\
& =  \mathds 1_{\{q^0 \phi(k) < \bar q\}} \left(\Lambda^1 \left( \Delta^{k,0}\left( \bar z^k, q\right) \right) \left( e^{\eta^0 z^{0,k,0} }y' - y \left(1 + \eta^0 \frac{1 - e^{-\gamma^0 \left(z^{0,k,0} + \Delta^{k,0}\left( \bar z^k, q\right) - \beta \left(\Delta^{k,0}\left( \bar z^k, q\right) - \underline{\Delta^{k}}\left( \bar z^k, q\right) \right) \right)}}{\gamma^0} \right) \right) \right.\\
&\  \left. + \Lambda^0 \left( \Delta^{k,0}\left( \bar z^k, q\right) \right) \left( e^{\eta^0 \left(z^{0,k,0} - c^0 \right)}y' - y \left(1 + \eta^0 \frac{1 - e^{-\gamma^0 \left(z^{0,k,0} + \Delta^{k,0}\left( \bar z^k, q\right) - \beta \left(\Delta^{k,0}\left( \bar z^k, q\right) - \underline{\Delta^{k}}\left( \bar z^k, q\right) \right) \right)}}{\gamma^0} \right) \right) \right)\\
& \ + \mathds 1_{\{q^1 \phi(k) < \bar q\}} \left(\Lambda^1 \left( \Delta^{k,1}\left( \bar z^k, q\right) \right) \left( e^{\eta^0 z^{0,k,1,1} }y'' -y\left(1 +\eta^0 \frac{1 - e^{-\gamma^0 \left(z^{0,k,1,1} \right)}}{\gamma^0}\right)  \right) \right.\\
&\  \left. + \Lambda^0 \left( \Delta^{k,1}\left( \bar z^k, q\right) \right) \left( e^{\eta^0 \left(z^{0,k,1,0} - c^0 \right)}y'' - y\left(1 + \eta^0 \frac{1 - e^{-\gamma^0 \left(z^{0,k,1,0}  \right)}}{\gamma^0}  \right) \right) \right),
\end{align*}
with $\bar z^k = \left(z^{0,k,0}, z^{1,k,1} \right)$, and the function $G^{0,k}:\mathbb R \times \bar{\mathcal Q}^2  \times \mathbb R \times \mathbb R\rightarrow \mathbb R$ such that
$$G^{0,k} \left(z^{1,k,1},q,  y,y' \right) = \underset{z^{0,k,0}, z^{0,k,1,:} \in \mathbb R \times \mathbb R^2}{\sup} g^{0,k} \left(z^{0,k,0}, z^{0,k,1,:}, z^{1,k,1},q,  y,y' \right).$$

We consider the following system of PDEs:
\begin{align}\label{10PDEsystem}
\begin{cases}
&0 = \partial_t \hat v_0(t,q) + G^{0,S}(q^0)\\
& \qquad \qquad+ \underset{k\in \{b,a\}}{\sum} G^{0,k} \Bigg( \check \zeta^{1,k,1}\Big(\hat v_1(t,q) , \hat v_1(t,q + \phi(k)e^1) \Big), q, \hat v_0(t,q), \hat v_0(t,q + \phi(k)e^0), \hat v_0(t,q + \phi(k)e^1)\Bigg),\\
&0 = \partial_t \hat v_1(t,q) - \frac{\gamma^1 \sigma^2}{2} (q^1)^2 +\mathcal H^1 \left( \check \zeta \bigg(q, \Big(\hat v_l(t,q) \Big)_{l = 0,1}, \Big( \hat v_l(t,q + \phi(k)e^i) \Big)_{\substack{k = b,a\\ l, i =0,1 }} \bigg) ,q \right),\\
&\hat v_0(T,q) = -1, \quad \hat v_1(T,q) = 0, 
\end{cases}
\end{align}
$\forall (t,q) \in [0,T) \times \bar{\mathcal Q}^2$,
where
\begin{align*}
\check \zeta &\bigg(q, \Big(\hat v_l(t,q) \Big)_{l = 0,1}, \Big( \hat v_l(t,q + \phi(k)e^i) \Big)_{\substack{k = b,a\\ l, i =0,1 }} \bigg)\\
&= \left(\check \zeta^0 \bigg(q,\hat v_0(t,q) , \Big( \hat v_0(t,q + \phi(k)e^i) \Big)_{\substack{k = b,a\\  i =0,1 }} \bigg), \check \zeta^1 \bigg(\hat v_1(t,q) , \Big( \hat v_1(t,q + \phi(k)e^i) \Big)_{\substack{k = b,a\\ i =0,1 }} \bigg) \right) \in \mathcal R^2,    
\end{align*}
with $\check \zeta^{0,k,0,0} = \check \zeta^{0,k,0,1} =: \check \zeta^{0,k,0}$ for $k\in \{b,a\}$, and
\begin{align*}
    & \check \zeta^{0,k,0} \bigg(q, \Big(\hat v_l(t,q) , \hat v_l(t,q + \phi(k)e^l) \Big)_{l = 0,1}  \bigg)\\
    & = \begin{cases}
    \tilde \zeta^{0,k,0} \left( \hat v_0(t,q) ,  \hat v_0(t,q + \phi(k)e^0)\right) \\
    \quad \text{if} \  \Delta^{k,0}\left(\tilde \zeta^{k} \left(\Big(\hat v_l(t,q) ,\hat v_l(t,q + \phi(k)e^l) \Big)_{l = 0,1} \right), q \right) \le \Delta^{k,1} \left(\tilde \zeta^{k} \left( \Big(\hat v_l(t,q) , \hat v_l(t,q + \phi(k)e^l) \Big)_{l = 0,1}\right), q\right),\\
    \tilde \zeta^{0,k,0}_\beta \left( \hat v_0(t,q) ,  \hat v_0(t,q + \phi(k)e^0)\right) \  \text{otherwise},
    \end{cases}\\
    & \check \zeta^{0,k,1,0} \bigg(\hat v_0(t,q), \hat v_0(t,q + \phi(k)e^1) \bigg) = \frac{\eta^0}{\eta^0 + \gamma^0} c^0 + \frac{1}{\eta^0 + \gamma^0} \log \left(\frac{\hat v_0(t,q)}{\hat v_0(t,q + \phi(k)e^1)} \right),\\
    & \check \zeta^{0,k,1,1} \bigg(\hat v_0(t,q), \hat v_0(t,q + \phi(k)e^1) \bigg) = \frac{1}{\eta^0 + \gamma^0} \log \left(\frac{\hat v_0(t,q)}{\hat v_0(t,q + \phi(k)e^1)} \right),\\
    & \check \zeta^{0,S}(q^0) = -\frac{\gamma^0}{\gamma^0 + \eta^0}q^0,
\end{align*}
and for $j\in\{0,1\}$
\begin{align*}
     & \check \zeta^{1,k,1,j} \bigg(\hat v_1(t,q) ,  \hat v_1(t,q + \phi(k)e^1) \bigg) = \hat v_1(t,q +\phi(k) e^1) - \hat v_1(t,q ),\\
    & \check \zeta^{1,k,0,j} \bigg(\hat v_1(t,q) ,  \hat v_1(t,q + \phi(k)e^0) \bigg) = \hat v_1(t,q +\phi(k) e^0) - \hat v_1(t,q ),\\
    & \check \zeta^{1,S} = 0.
\end{align*}

In the above, we have
\begin{align*}
&\tilde \zeta^{0,k,0} \left( \hat v_0(t,q) ,  \hat v_0(t,q + \phi(k)e^0)\right)\\
&\ = \frac{1}{\eta^0} \left[ \log \left( \frac{\hat v_0(t,q)}{\hat v_0(t,q + \phi(k)e^0)}\right) + \log \left(1 - \frac{\sigma^2 \gamma^0 \eta^0}{\left(  \kappa + \sigma \gamma^0\right)\left(  \kappa + \sigma \eta^0\right)} \right) + \log \left( \frac{\sum_{j = 0}^1 A^j e^{- \frac{ \kappa}{\sigma}c^j}}{A^1 e^{- \frac{ \kappa}{\sigma}c^1} + A^0 e^{- \left(\frac{ \kappa}{\sigma} + \eta^0\right)c^0}} \right) \right],\\
&\tilde \zeta^{0,k,0}_\beta \left( \hat v_0(t,q) ,  \hat v_0(t,q + \phi(k)e^0)\right)\\
&\ = \frac{1}{\eta^0} \Bigg[ \log \left( \frac{\hat v_0(t,q)}{\hat v_0(t,q + \phi(k)e^0)}\right)\\
&\qquad \qquad \qquad  + \log \left(1 - \frac{(1-\beta)^2\sigma^2 \gamma^0 \eta^0}{\left(  \kappa + (1-\beta)\sigma \gamma^0\right)\left(  \kappa + (1-\beta)\sigma \eta^0\right)} \right) + \log \Bigg( \frac{\sum_{j = 0}^1 A^j e^{- \frac{ \kappa}{\sigma}c^j}}{A^1 e^{- \frac{ \kappa}{\sigma}c^1} + A^0 e^{- \left(\frac{ \kappa}{\sigma} + \eta^0\right)c^0}} \Bigg) \Bigg],    
\end{align*}
and
\begin{align*}
\tilde \zeta^{k} \left(\Big(\hat v_l(t,q) ,\hat v_l(t,q + \phi(k)e^l) \Big)_{l = 0,1} \right) = \left(\tilde \zeta^{0,k,0} \left( \hat v_0(t,q) ,  \hat v_0(t,q + \phi(k)e^0)\right), \check \zeta^{1,k,1} \big(\hat v_1(t,q) ,  \hat v_1(t,q + \phi(k)e^1) \big) \right).    
\end{align*}

The following lemma is a straightforward application of Cauchy-Lipschitz Theorem:

\begin{lemme}\label{solE10}
There exists a unique solution $\left(\hat v_0, \hat v_1 \right)$ to the system of PDEs \eqref{10PDEsystem}. 
\end{lemme}

This allows to conclude with the classical verification theorem below.

\begin{thm}\label{verif10}
Let us consider the solution $\left(\hat v_0, \hat v_1 \right)$ to the system of PDEs \eqref{10PDEsystem} given by \Cref{solE10}. Then the contract of exchange $E_0$ given by
\begin{align*}
    &Z^{0,S}_t = \check \zeta^{0,S}(Q^0_t),\\
    &Z^{0,k,0}_t = \check \zeta^{0,k,0} \bigg(Q_t, \Big(\hat v_l(t,Q_t) , \hat v_l(t,Q_t + \phi(k)e^l) \Big)_{l = 0,1}  \bigg),\\
    &Z^{0,k, 1,j}_t = \check \zeta^{0,k,1,j} \bigg(\hat v_0(t,Q), \hat v_0(t,Q + \phi(k)e^1) \bigg)\quad \forall j \in\{0,1\}, \quad \forall k\in\{b,a\},
\end{align*}
solves the Problem \eqref{pb10}. Moreover, given the corresponding process $Z^0$, there  exists a unique solution to BSDE \eqref{0BSDE}, which is given by
\begin{align*}
    Z^{1,S}_t = 0,\quad Z^{1,k,1,j}_t = \check \zeta^{1,k,1} \bigg(\hat v_1(t,Q_t) ,  \hat v_1(t,Q_t + \phi(k)e^1) \bigg),\quad Z^{1,k,0,j}_t = \check \zeta^{1,k,0,j} \bigg(\hat v_1(t,Q_t) ,  \hat v_1(t,Q_t + \phi(k)e^0) \bigg),
\end{align*}
$\forall j \in\{0,1\}$, $\forall k\in\{b,a\}$, $\forall t \in [0,T]$.\\

In particular, the representation of the null contract of exchange $E_1$ is given by the corresponding process $Z^1$.
\end{thm}

\subsection{Incentives for both market makers}\label{ssec:2-incentive} 

Finally, we study the situation in which both exchanges propose a contract to their respective market maker.

\hs

The optimization problem of exchange $E_i$ is given by
\begin{align}
    V_{E_i}(\xi^j) := \!\!\!\!\underset{\xi^i \in {\Xi'}^i(\xi^j)}{\sup} \mathbb E^{\hat \delta(\xi) } \left[-e^{-\eta^i \left(c^i\left( N^{b,i,i}_T + N^{b,j,i}_T + N^{a,i,i}_T + N^{a,j,i}_T \right) - \xi^i \right)} \right].
\end{align}

\hs

As before, in view of Remark \ref{val_func_mm_0} the optimal $Y^0_0$ and $Y^1_0$ are known and can be factored out of the equation, so that we only optimize over the processes $Z$.\\

As for exchange $E_0$, let us define for exchange $E_1$ the function $g^{1,S}: \mathbb R \times \bar{\mathcal Q}  \rightarrow \mathbb R$ such that
$$g^{1,S}\left(z^{1,S}, q^1 \right) = \frac{\eta^1 \sigma^2}{2} \left( \gamma^1 \left(z^{1,S} + q^1\right)^2 + \eta^1 \left(z^{1,S}\right)^2 \right) \quad \forall \left(z^{1,S}, q^1 \right) \in \mathbb R \times \bar{\mathcal Q},$$

and the function $G^{1,S}:\bar{\mathcal Q}  \rightarrow \mathbb R$ given by 
$$G^{1,S}(q^1) = \underset{z^{1,S}\in \mathbb R}{\sup} g^{1,S}\left(z^{1,S}, q^1 \right).$$

Let us also define for $k\in\{b,a\}$ the function $g^{1,k}:\mathbb R \times \mathbb R^2 \times \mathbb R \times \bar{\mathcal Q}^2  \times \mathbb R \times \mathbb R \times \mathbb R\rightarrow \mathbb R$ such that
\begin{align*}
& g^{1,k} \left(z^{1,k,1}, z^{1,k,0,:}, z^{0,k,0},q,  y,y', y'' \right)\\
& =  \mathds 1_{\{q^1 \phi(k) < \bar q\}} \left(\Lambda^0 \left( \Delta^{k,1}\left( \bar z^k, q\right) \right) \left( e^{\eta^1 z^{1,k,1} }y' - y \left(1 + \eta^1 \frac{1 - e^{-\gamma^1 \left(z^{1,k,1} + \Delta^{k,1}\left( \bar z^k, q\right) - \beta \left(\Delta^{k,1}\left( \bar z^k, q\right) - \underline{\Delta^{k}}\left( \bar z^k, q\right) \right) \right)}}{\gamma^1} \right) \right) \right.\\
&\  \left. + \Lambda^1 \left( \Delta^{k,1}\left( \bar z^k, q\right) \right) \left( e^{\eta^1 \left(z^{1,k,1} - c^1 \right)}y' - y \left(1 + \eta^1 \frac{1 - e^{-\gamma^1 \left(z^{1,k,1} + \Delta^{k,1}\left( \bar z^k, q\right) - \beta \left(\Delta^{k,1}\left( \bar z^k, q\right) - \underline{\Delta^{k}}\left( \bar z^k, q\right) \right) \right)}}{\gamma^1} \right) \right) \right)\\
& \ + \mathds 1_{\{q^0 \phi(k) < \bar q\}} \left(\Lambda^0 \left( \Delta^{k,0}\left( \bar z^k, q\right) \right) \left( e^{\eta^1 z^{1,k,0,0} }y'' -y\left(1 +\eta^1 \frac{1 - e^{-\gamma^1 \left(z^{1,k,0,0} \right)}}{\gamma^1}\right)  \right) \right.\\
&\  \left. + \Lambda^1 \left( \Delta^{k,0}\left( \bar z^k, q\right) \right) \left( e^{\eta^1 \left(z^{1,k,0,1} - c^1 \right)}y'' - y\left(1 + \eta^1 \frac{1 - e^{-\gamma^1 \left(z^{1,k,0,1}  \right)}}{\gamma^1}  \right) \right) \right),
\end{align*}
with $\bar z^k = \left(z^{0,k,0}, z^{1,k,1} \right)$, and the function $G^{1,k}:\mathbb R \times \bar{\mathcal Q}^2  \times \mathbb R \times \mathbb R\rightarrow \mathbb R$ such that
$$G^{1,k} \left(z^{0,k,0},q,  y,y' \right) = \underset{z^{1,k,1}, z^{1,k,0,:} \in \mathbb R \times \mathbb R^2}{\sup} g^{1,k} \left(z^{1,k,1}, z^{1,k,0,:}, z^{0,k,0},q,  y,y' \right).$$

We consider the following system of PDEs for $m=0,1$:
\begin{align}\label{11PDEsystem}
\begin{cases}
&0 = \partial_t v_m(t,q) + G^{m,S}(q^m)\\
& \quad+\!\! \underset{k\in \{b,a\}}{\sum} \!\!G^{m,k} \Bigg(   \zeta^{1-m,k,1-m}\Big(  v_{1-m}(t,q) ,   v_{1-m}(t,q + \phi(k)e^{1-m}) \Big), q,   v_m(t,q),   v_m(t,q + \phi(k)e^m),   v_m(t,q + \phi(k)e^{1-m})\Bigg),\\
&  v_0(T,q) = -1, \quad   v_1(T,q) = -1, 
\end{cases}
\end{align}
$\forall (t,q) \in [0,T) \times \bar{\mathcal Q}^2$,
where
\begin{align*}
 \zeta &\bigg(q, \Big( v_l(t,q) \Big)_{l = 0,1}, \Big(  v_l(t,q + \phi(k)e^i) \Big)_{\substack{k = b,a\\ l, i =0,1 }} \bigg)\\
&= \left( \zeta^0 \bigg(q,v_0(t,q) , \Big(  v_0(t,q + \phi(k)e^i) \Big)_{\substack{k = b,a\\  i =0,1 }} \bigg), \zeta^1 \bigg( v_1(t,q) , \Big(  v_1(t,q + \phi(k)e^i) \Big)_{\substack{k = b,a\\ i =0,1 }} \bigg) \right) \in \mathcal R^2,    
\end{align*}
with $ \zeta^{0,k,0,0} =  \zeta^{0,k,0,1} =: \zeta^{0,k,0}$ for $k\in \{b,a\}$, and
\begin{align*}
    &  \zeta^{0,k,0} \bigg(q, \Big( v_l(t,q) ,  v_l(t,q + \phi(k)e^l) \Big)_{l = 0,1}  \bigg)\\
    & = \begin{cases}
    \tilde \zeta^{0,k,0} \left(  v_0(t,q) ,   v_0(t,q + \phi(k)e^0)\right) \\
    \quad \text{if} \  \Delta^{k,0}\left(\bar \zeta^{k} \left(\Big( v_l(t,q) , v_l(t,q + \phi(k)e^l) \Big)_{l = 0,1} \right), q \right) \le \Delta^{k,1} \left(\bar\zeta^{k} \left( \Big( v_l(t,q) ,  v_l(t,q + \phi(k)e^l) \Big)_{l = 0,1}\right), q\right),\\
    \tilde \zeta^{0,k,0}_\beta \left(  v_0(t,q) ,  v_0(t,q + \phi(k)e^0)\right) \  \text{otherwise},
    \end{cases}\\
    & \zeta^{0,k,1,0} \bigg( v_0(t,q),  v_0(t,q + \phi(k)e^1) \bigg) = \frac{\eta^0}{\eta^0 + \gamma^0} c^0 + \frac{1}{\eta^0 + \gamma^0} \log \left(\frac{ v_0(t,q)}{v_0(t,q + \phi(k)e^1)} \right),\\
    &  \zeta^{0,k,1,1} \bigg( v_0(t,q),  v_0(t,q + \phi(k)e^1) \bigg) = \frac{1}{\eta^0 + \gamma^0} \log \left(\frac{ v_0(t,q)}{ v_0(t,q + \phi(k)e^1)} \right),\\
    & \zeta^{0,S}(q^0) = -\frac{\gamma^0}{\gamma^0 + \eta^0}q^0,
\end{align*}
and $ \zeta^{1,k,1,0} =  \zeta^{1,k,1,1} =: \zeta^{1,k,1}$ for $k\in \{b,a\}$, and
\begin{align*}
    &  \zeta^{1,k,1} \bigg(q, \Big( v_l(t,q) ,  v_l(t,q + \phi(k)e^l) \Big)_{l = 0,1}  \bigg)\\
    & = \begin{cases}
    \tilde \zeta^{1,k,1} \left(  v_1(t,q) ,   v_1(t,q + \phi(k)e^1)\right) \\
    \quad \text{if} \  \Delta^{k,1}\left(\bar \zeta^{k} \left(\Big( v_l(t,q) , v_l(t,q + \phi(k)e^l) \Big)_{l = 0,1} \right), q \right) \le \Delta^{k,0} \left(\bar\zeta^{k} \left( \Big( v_l(t,q) ,  v_l(t,q + \phi(k)e^l) \Big)_{l = 0,1}\right), q\right),\\
    \tilde \zeta^{1,k,1}_\beta \left(  v_1(t,q) ,  v_1(t,q + \phi(k)e^1)\right) \  \text{otherwise},
    \end{cases}\\
    & \zeta^{1,k,0,1} \bigg( v_1(t,q),  v_1(t,q + \phi(k)e^0) \bigg) = \frac{\eta^1}{\eta^1 + \gamma^1} c^1 + \frac{1}{\eta^1 + \gamma^1} \log \left(\frac{ v_1(t,q)}{v_1(t,q + \phi(k)e^0)} \right),\\
    &  \zeta^{1,k,0,0} \bigg( v_1(t,q),  v_1(t,q + \phi(k)e^0) \bigg) = \frac{1}{\eta^1 + \gamma^1} \log \left(\frac{ v_1(t,q)}{v_1(t,q + \phi(k)e^0)} \right),\\
    & \zeta^{1,S}(q^1) = -\frac{\gamma^1}{\gamma^1 + \eta^1}q^1.
\end{align*}

In the above, we have
\begin{align*}
&\tilde \zeta^{0,k,0} \left(   v_0(t,q) ,    v_0(t,q + \phi(k)e^0)\right)\\
&\ = \frac{1}{\eta^0} \left[ \log \left( \frac{  v_0(t,q)}{  v_0(t,q + \phi(k)e^0)}\right) + \log \left(1 - \frac{\sigma^2 \gamma^0 \eta^0}{\left(  \kappa + \sigma \gamma^0\right)\left(  \kappa + \sigma \eta^0\right)} \right) + \log \left( \frac{\sum_{j = 0}^1 A^j e^{- \frac{ \kappa}{\sigma}c^j}}{A^1 e^{- \frac{ \kappa}{\sigma}c^1} + A^0 e^{- \left(\frac{ \kappa}{\sigma} + \eta^0\right)c^0}} \right) \right],
\end{align*}
\begin{align*}
&\tilde \zeta^{0,k,0}_\beta \left(   v_0(t,q) ,    v_0(t,q + \phi(k)e^0)\right)\\
&\ = \frac{1}{\eta^0} \Bigg[ \log \left( \frac{  v_0(t,q)}{  v_0(t,q + \phi(k)e^0)}\right)\\
&\qquad \qquad \qquad  + \log \left(1 - \frac{(1-\beta)^2\sigma^2 \gamma^0 \eta^0}{\left(  \kappa + (1-\beta)\sigma \gamma^0\right)\left(  \kappa + (1-\beta)\sigma \eta^0\right)} \right) + \log \Bigg( \frac{\sum_{j = 0}^1 A^j e^{- \frac{ \kappa}{\sigma}c^j}}{A^1 e^{- \frac{ \kappa}{\sigma}c^1} + A^0 e^{- \left(\frac{ \kappa}{\sigma} + \eta^0\right)c^0}} \Bigg) \Bigg],    
\end{align*}

\begin{align*}
&\tilde \zeta^{1,k,1} \left(   v_1(t,q) ,    v_1(t,q + \phi(k)e^1)\right)\\
&\ = \frac{1}{\eta^1} \left[ \log \left( \frac{  v_1(t,q)}{  v_1(t,q + \phi(k)e^1)}\right) + \log \left(1 - \frac{\sigma^2 \gamma^1 \eta^1}{\left(  \kappa + \sigma \gamma^1\right)\left(  \kappa + \sigma \eta^1\right)} \right) + \log \left( \frac{\sum_{j = 0}^1 A^j e^{- \frac{ \kappa}{\sigma}c^j}}{A^0 e^{- \frac{ \kappa}{\sigma}c^0} + A^1 e^{- \left(\frac{ \kappa}{\sigma} + \eta^1\right)c^1}} \right) \right],\\
&\tilde \zeta^{1,k,1}_\beta \left(   v_1(t,q) ,    v_1(t,q + \phi(k)e^1)\right)\\
&\ = \frac{1}{\eta^1} \Bigg[ \log \left( \frac{  v_1(t,q)}{  v_1(t,q + \phi(k)e^1)}\right)\\
&\qquad \qquad \qquad  + \log \left(1 - \frac{(1-\beta)^2\sigma^2 \gamma^1 \eta^1}{\left(  \kappa + (1-\beta)\sigma \gamma^1\right)\left(  \kappa + (1-\beta)\sigma \eta^1\right)} \right) + \log \Bigg( \frac{\sum_{j = 0}^1 A^j e^{- \frac{ \kappa}{\sigma}c^j}}{A^0 e^{- \frac{ \kappa}{\sigma}c^0} + A^1 e^{- \left(\frac{ \kappa}{\sigma} + \eta^1\right)c^1}} \Bigg) \Bigg],    
\end{align*}
and
\begin{align*}
\tilde \zeta^{k} \left(\Big( v_l(t,q) , v_l(t,q + \phi(k)e^l) \Big)_{l = 0,1} \right) = \left(\tilde \zeta^{0,k,0} \left(  v_0(t,q) ,   v_0(t,q + \phi(k)e^0)\right), \tilde \zeta^{1,k,1} \big( v_1(t,q) ,  v_1(t,q + \phi(k)e^1) \big) \right).    
\end{align*}

We have once again the following results:

\begin{lemme}\label{solE11}
There exists a unique solution $\left(v_0, v_1 \right)$ to the system of PDEs \eqref{11PDEsystem}. 
\end{lemme}

Finally, we state below the verification theorem.

\begin{thm}\label{verif11}
Let us consider the solution $\left( v_0,v_1 \right)$ to the system of PDEs \eqref{11PDEsystem} given by \Cref{solE11}. Then the contract of exchange $E_0$ given by
\begin{align*}
    &Z^{0,S}_t =  \zeta^{0,S}(Q^0_t),\\
    &Z^{0,k,0}_t = \zeta^{0,k,0} \bigg(Q_t, \Big(  v_l(t,Q_t) ,   v_l(t,Q_t + \phi(k)e^l) \Big)_{l = 0,1}  \bigg),\\
    &Z^{0,k, 1,j}_t = \zeta^{0,k,1,j} \bigg(  v_0(t,Q),   v_0(t,Q + \phi(k)e^1) \bigg)\quad \forall j \in\{0,1\},
\end{align*}
$\forall k\in\{b,a\}$, together with the contract of exchange $E_1$ given by
\begin{align*}
    &Z^{1,S}_t =  \zeta^{1,S}(Q^1_t),\\
    &Z^{1,k,1}_t = \zeta^{1,k,1} \bigg(Q_t, \Big(  v_l(t,Q_t) ,   v_l(t,Q_t + \phi(k)e^l) \Big)_{l = 0,1}  \bigg),\\
    &Z^{1,k, 0,j}_t = \zeta^{1,k,0,j} \bigg(  v_1(t,Q),   v_1(t,Q + \phi(k)e^0) \bigg)\quad \forall j \in\{0,1\},
\end{align*}
$\forall j \in\{0,1\}$, $\forall k\in\{b,a\}$, $\forall t \in [0,T]$, represent a Nash equilibrium of the Problem.\\
\end{thm}

\section{Economic results}

We analyse the effects of the interactions between the two platforms in different settings of incentive provision to market makers. The Table~\ref{tab:params} gives the value of the different parameters of the model. We considered a single trading day ($T=1$ day). All the equations are solved numerically using explicit Euler schemes.

\hs

\begin{table}[h]
\centering
\begin{tabular}{c c c c c c c c c} 
$\sigma$ & $\kappa$ & $A^0, A^1$ & $c^0, c_1$ & $\eta^0, \eta^1$ & $\gamma^0, \gamma^1$ & $\beta$ & $\bar q$ & $\delta_\infty$      \\ \hline
& &  & &  & &  &  &   \\ 
$\$ \cdot\text{day}^{-\frac 12}$ & $\$^{-1}\cdot \text{day}^{-\frac 12}$ & day$^{-1}$ & $\$$ &  $\$^{-1}$   & $\$^{-1}$  &  --     & nb share & $\$$   \\ \hline
& &  &  & &  &  &&   \\ 
1.2 & 8 & 100 &  $10^{-5}$ & 0.1 &  0.01 & 0.6 & 5  &  10   \\ \hline
\end{tabular}
\caption{Parameter values along with their corresponding units.}
\label{tab:params}
\end{table}

\subsection{Main result}

Figures~\ref{fig:ExchangeV} presents the value functions of the platforms as a function of inventories in the different possible configuration of competition between them: absence of any contract, presence of a single incentive contract (from Exchange~0) and both platforms providing incentives to their market maker. The vignette~(a) clearly shows the large gap in value for Exchange~0  that results from moving unilaterally from an absence of contract to the implementation of an incentive mechanism. Besides, we observe only a marginal variation when Exchange~1 also implements an optimal incentive contract. By comparison, vignette (b) shows that Exchange~1 largely benefits from the implementation of an incentive contract from her competitor, even if Exchange~1 does not provide incentives herself to her own market maker. Finally, vignette~(c) shows the sum of the value functions of the two exchanges as a function of the inventory of exchange~0. It can be observed that the market would largely benefit from moving from an absence of market-making incentive mechanisms to the presence of at least one contract. But the second contract brings a much smaller incremental benefit. This phenomenon is also illustrated in vignette (d) which provides the values of the two exchanges as a function of the risk-aversion parameters of the market makers when they are equal. It can be observed that for any level of risk-aversion of the market makers the unilateral provision of incentive by the Exchange~0 largely increase the utility of Exchange~1 (she moves from approx. $-30$ to $-13$). But, Exchange~1 only marginally benefits from implementing herself in her turn a mechanism (she then moves from $-13$ to $-12$).

\begin{figure}[!htb]
    \begin{subfigure}{0.45\textwidth}
        \caption{}
        \includegraphics[width=\textwidth]{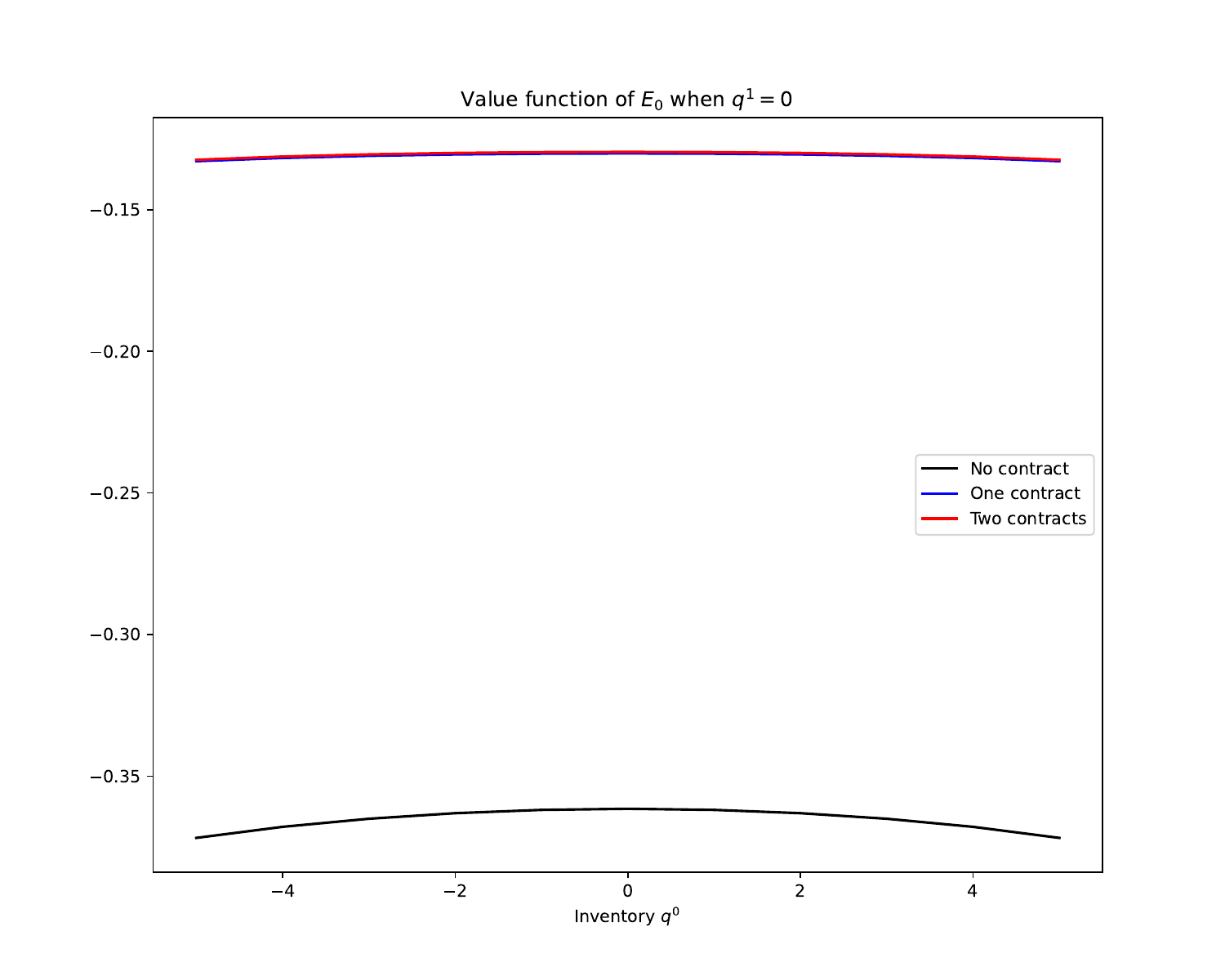}
    \end{subfigure}
    \begin{subfigure}{0.45\textwidth}
        \caption{}
        \includegraphics[width=\textwidth]{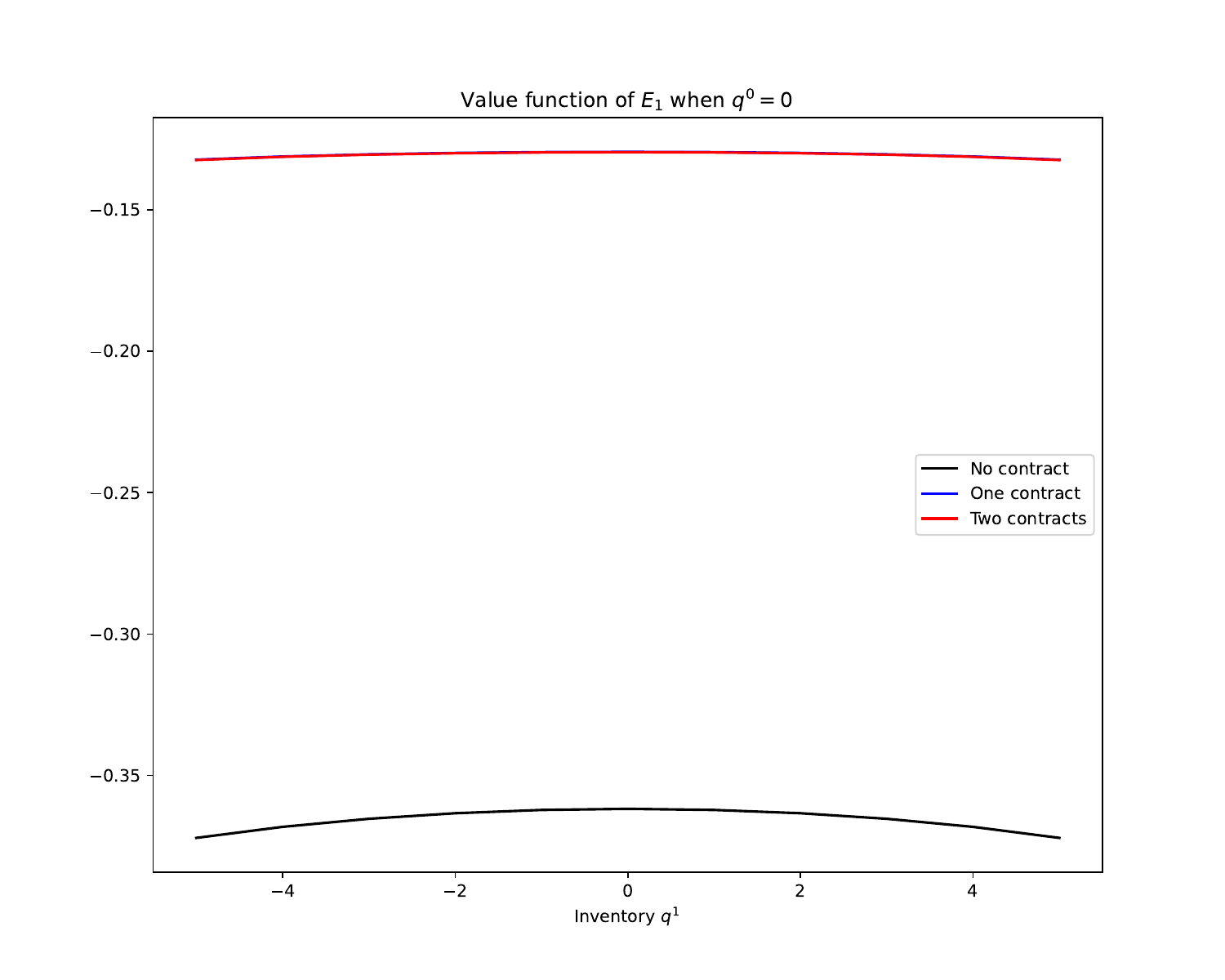}
    \end{subfigure}\\
    \begin{subfigure}{0.45\textwidth}
        \caption{}
        \includegraphics[width=\textwidth]{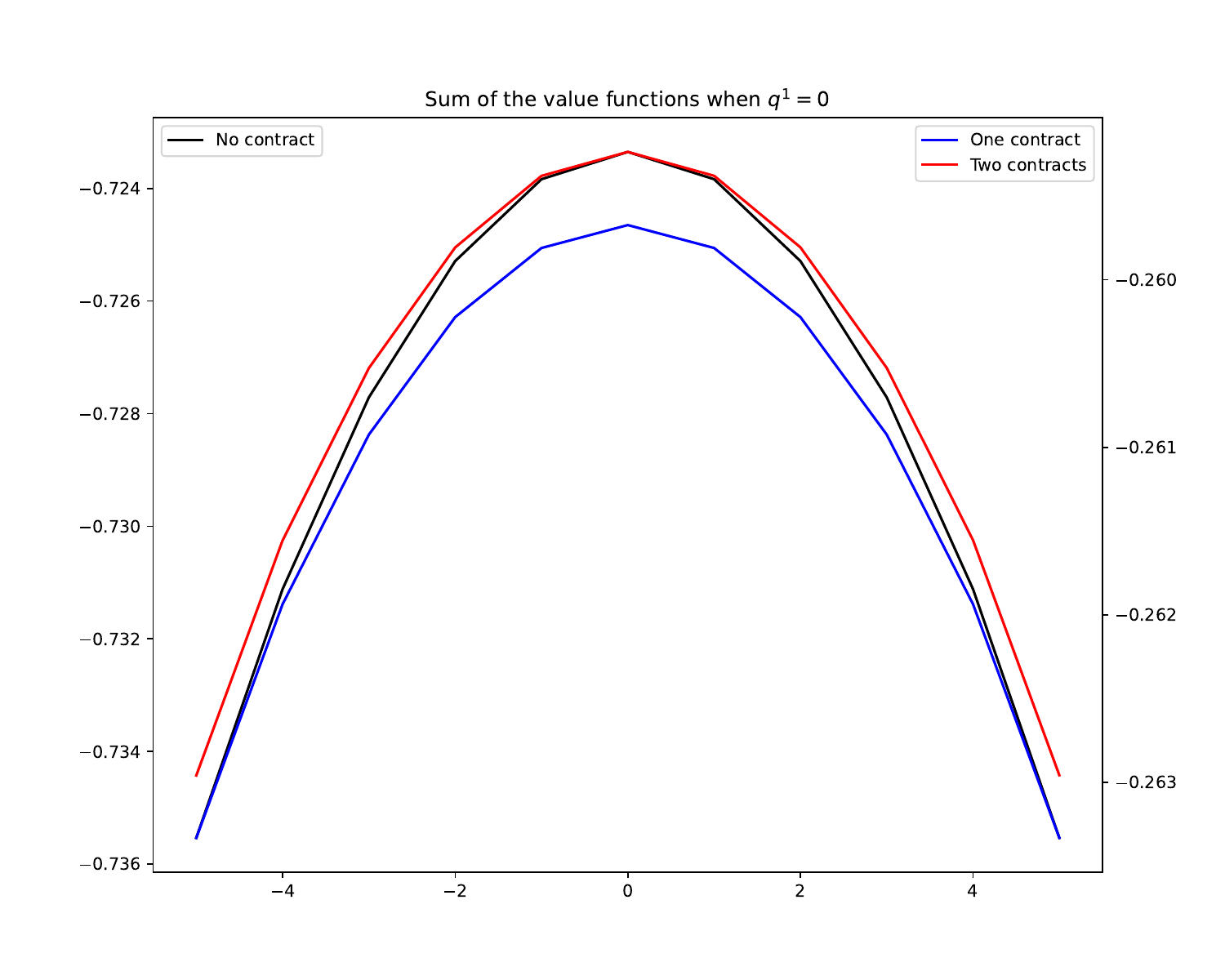}
    \end{subfigure}
    \begin{subfigure}{0.45\textwidth}
        \caption{}
        \includegraphics[width=0.9\linewidth]{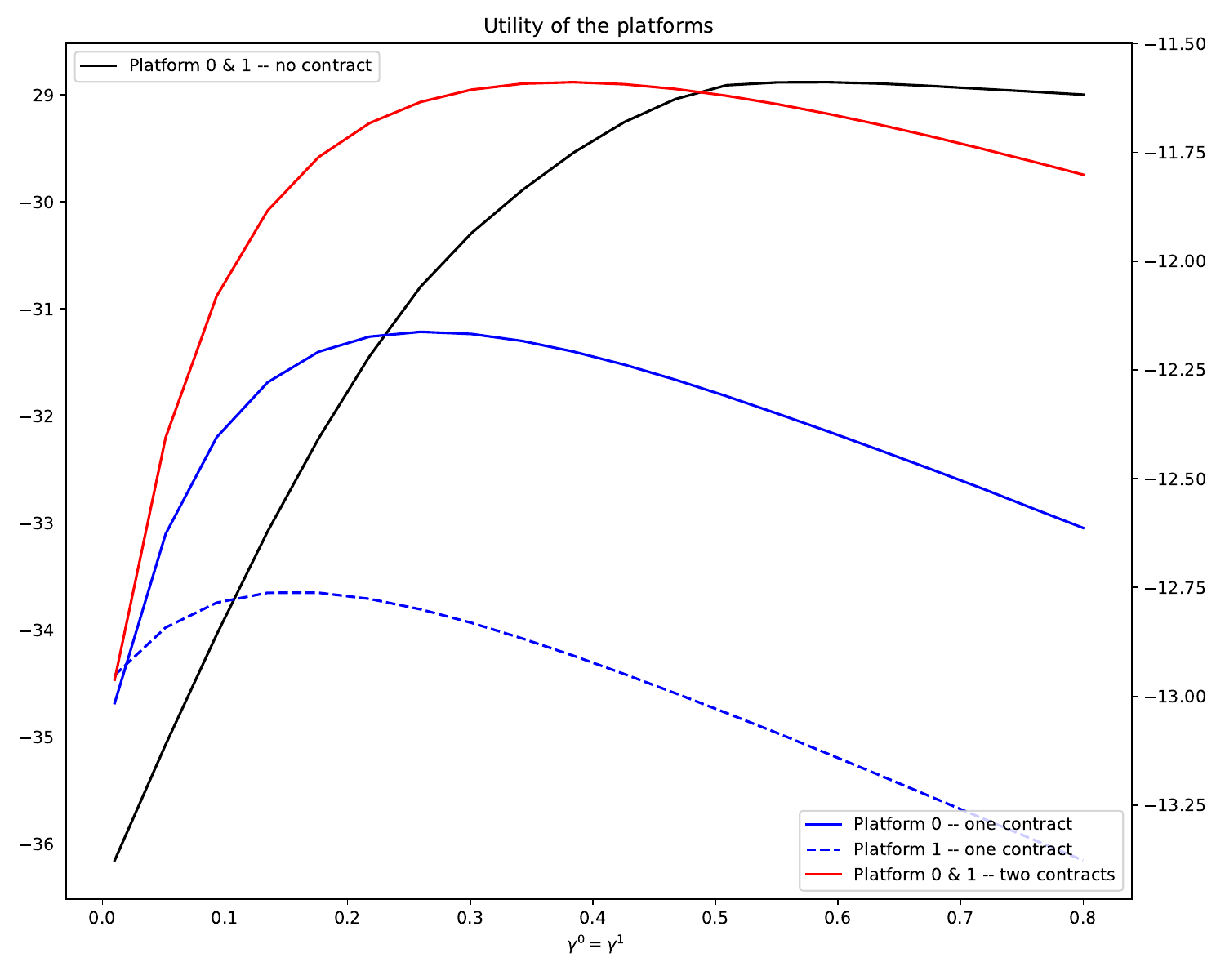}
    \end{subfigure}
    \caption{ {\footnotesize Value functions of the exchanges. (a) Value function of Exchange 0 as a function of her inventory $q^0$ when the inventory $q^1$  of Exhange~1 is zero (b) Value function of Exchange 1 as a function of her inventory $q^1$ when the inventory $q^0$  of Exhange~0 is zero (c) Sum of the value functions of exchanges as a function of $q^0$ when $q^1=0$; (d) Value function of the exchanges as a function of the risk-aversion parameters of the market makers when their are equal. Right axis: absence of contract, Left axis: presence of at least one incentive contract.}}
\label{fig:ExchangeV}
\end{figure}


\hs

This result can be understood in the following way. Incentives provided by Exchange~0 to her market-maker $M_0$ induces more aggressive quotes, which leads to an increase in market orders to $M_0$, which increases the competitive pressure to $M_1$ who is compelled to adjust its own quotes to maintain its level of market orders. This variation of bid quotes is illustrated in Figure~\ref{fig:quotes} in two situations: vignette (a) when only the risk-aversion parameter of the market maker~0 increases (exchanges faces heterogeneous market makers) and vignette (b) when both risk-aversion parameters increases (exchanges face similar market makers). In both situations bid quotes are given for an empty inventory, and are thus equal between market makers. One can observe that in both situation and for any level of risk-aversion, a large reduction in the quotes when moving from an absence of contract to the implementation of an incentive mechanism, and the small decrease when moving from one to two incentives mechanisms. All the liquidity gain is generated from the first exchange implementing an incentive contract. The other exchange benefits from the action of her competitor without doing herself anything. Indeed, when the first exchange incentivizes her maker maker, he provides better quotes and thus becomes more competitive. The other maker maker has no alternative than reducing his own quotes to remain competitive. This result is an effect of {\em competitiveness spill over}.

\hs

The presence of this competitiveness spill-over effect makes incentive provision a public good. Both platforms would personally benefit from the incentives provided by the other exchange to her market maker. But, because the provision of incentive is costly (the platform has to pay the reservation utility of the market maker), an exchange may benefit from not incentivizing her market maker and letting the other platform pay for incentive provision. It is a free-rider problem. As often in the provision of public good, the consequence may be dramatic at equilibrium, as the Nash equilibrium results in no incentive provision by any exchange.

\subsection{Quotes}\label{ssec:quotes}

This section illustrates the effect of competition on the behavior of the market makers in the absence of contracts and shows how our model produces results in line with the economic intuition. In what follows, we only study the optimal bid quotes of the market makers, as bid and ask are symmetric.

\hs

\begin{figure}[!htb]
\begin{subfigure}{.45\linewidth}
        \caption{}
  \includegraphics[width=\linewidth]{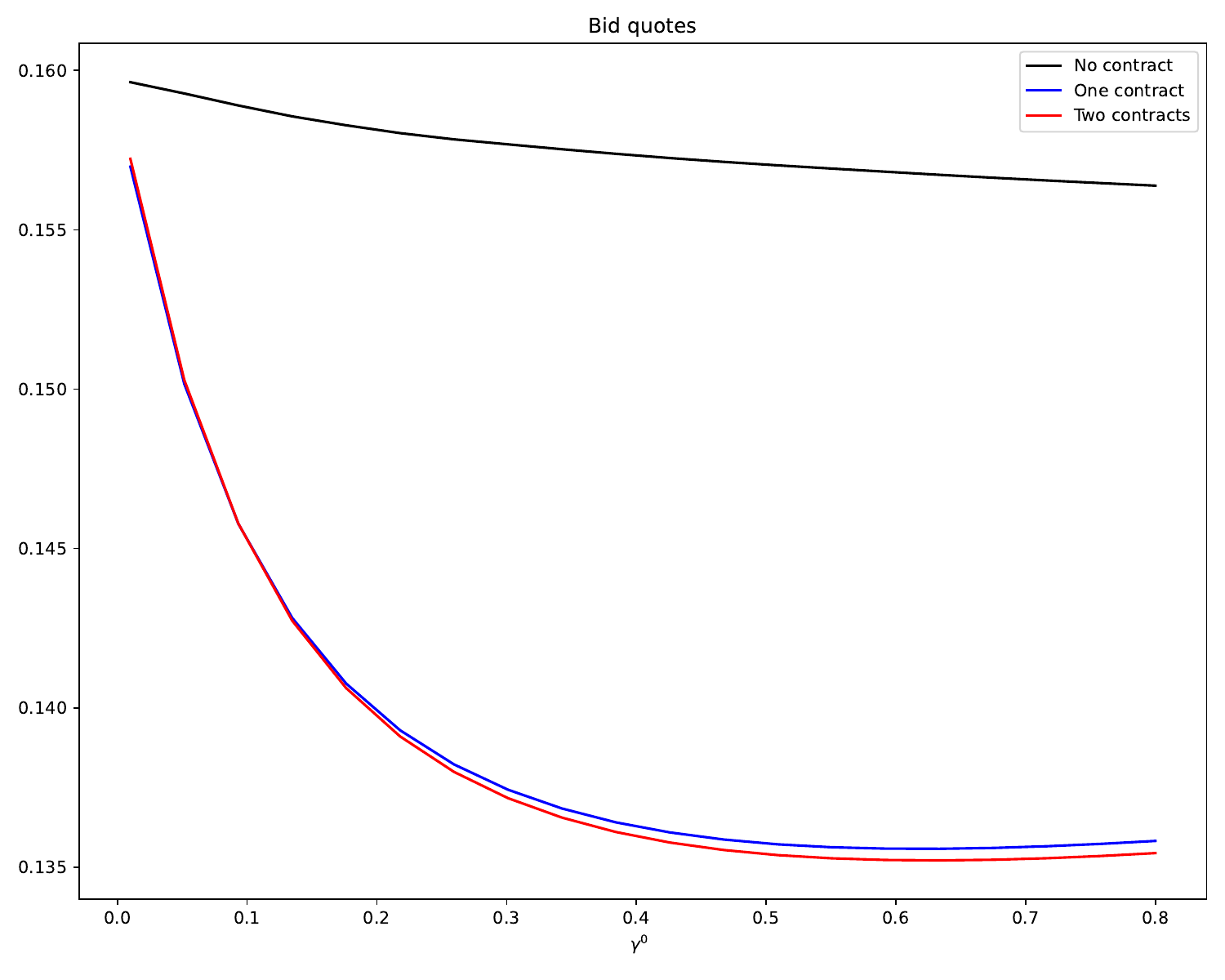}
\end{subfigure}\hfill 
\begin{subfigure}{.45\linewidth}
        \caption{}
  \includegraphics[width=\linewidth]{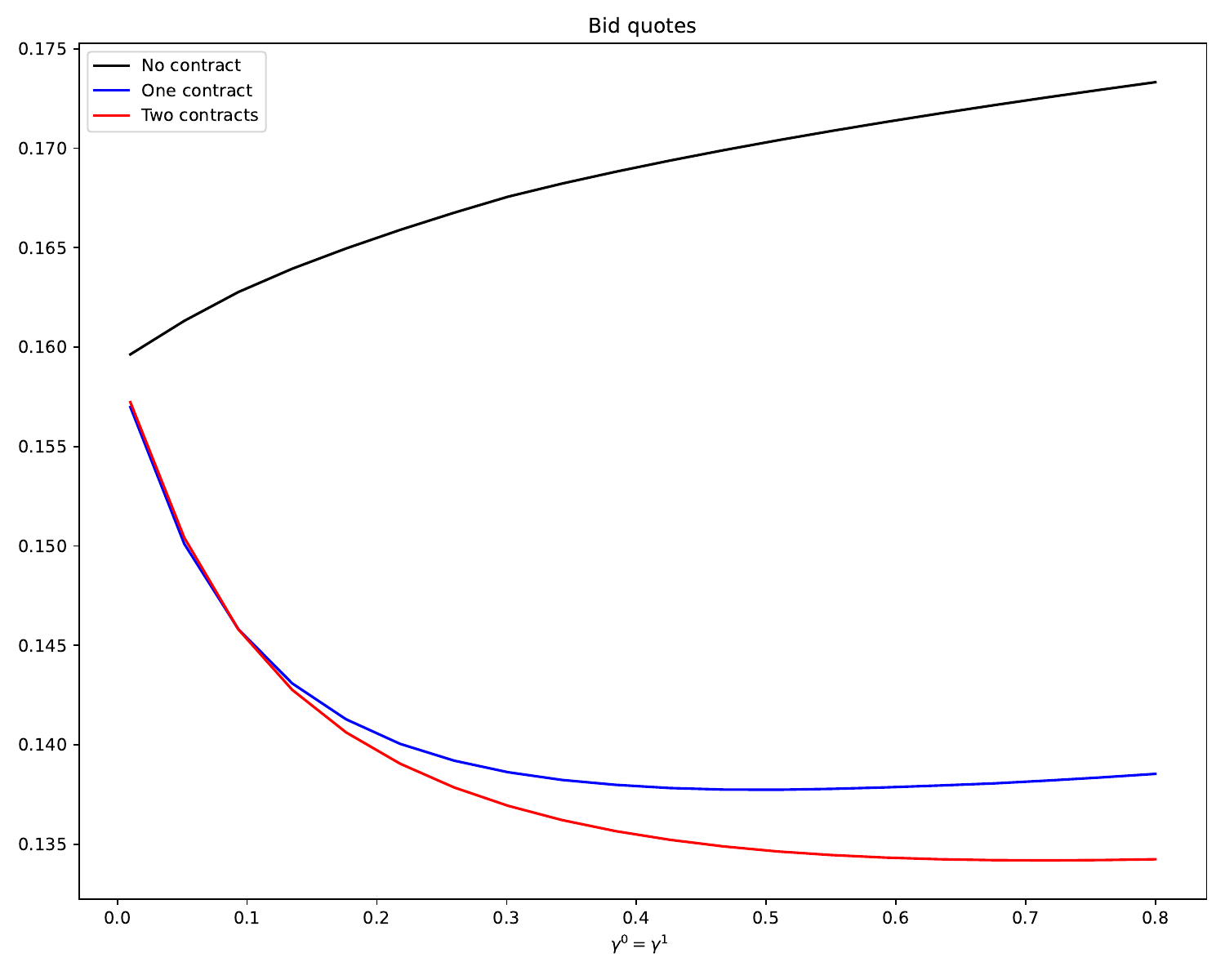}
\end{subfigure}
\caption{Bid quote of the market makers as a function of $\gamma^0$ when both inventories are zero.}
\label{fig:quotes}
\end{figure}

On \Cref{fig:deltas}, the left side shows the optimal bid quote of market maker $M_0$ when market maker $M_1$ has an empty inventory. In the absence of contracts, we observe that as long as its inventory is negative, the optimal bid quotes evolves in line with the intuition, as it increases with inventory. This is because, with a negative inventory, market maker $M_0$ wants to buy and therefore consistently proposes a smaller mid-to-bid quote than market maker $M_1$, which has an empty inventory. As soon as its inventory goes above 0, market maker $M_1$ becomes the leader, in the sense that she has more incentives to buy than market maker $M_0$ who is already long. What we observe here is that, with a reasonably long inventory, market maker $M_0$ will try to keep up and align its quotes with market maker $M_1$ in order to remain competitive, while, on the other side, market maker $M_1$ who is conscious of its advantage over market maker $M_0$ will become more aggressive. This is confirmed by the right side of \Cref{fig:deltas}. However, we see that, at some point, when the inventory of market maker $M_0$ becomes to long, she stops trying to compete because it represents too much risk for her.

\hs

We now turn to the cases where exchanges offer contracts to their respective market makers. When only exchange $E_0$ provides a contract, the quotes of both market makers decrease around inventory zero -- the region where she spends most of her time -- which significantly increases effective liquidity. In the two-contract case, we observe that both market makers reduce their spreads substantially. The quotes become much more stable with respect to inventory and notably lower than the benchmark case without contracts around inventory zero. Both market makers accept more risk by reducing their bid quote with a positive inventory, as they are now compensated by their exchange for providing tighter quotes.

\begin{figure}[!h]
\begin{subfigure}{.5\linewidth}
  \includegraphics[width=\linewidth]{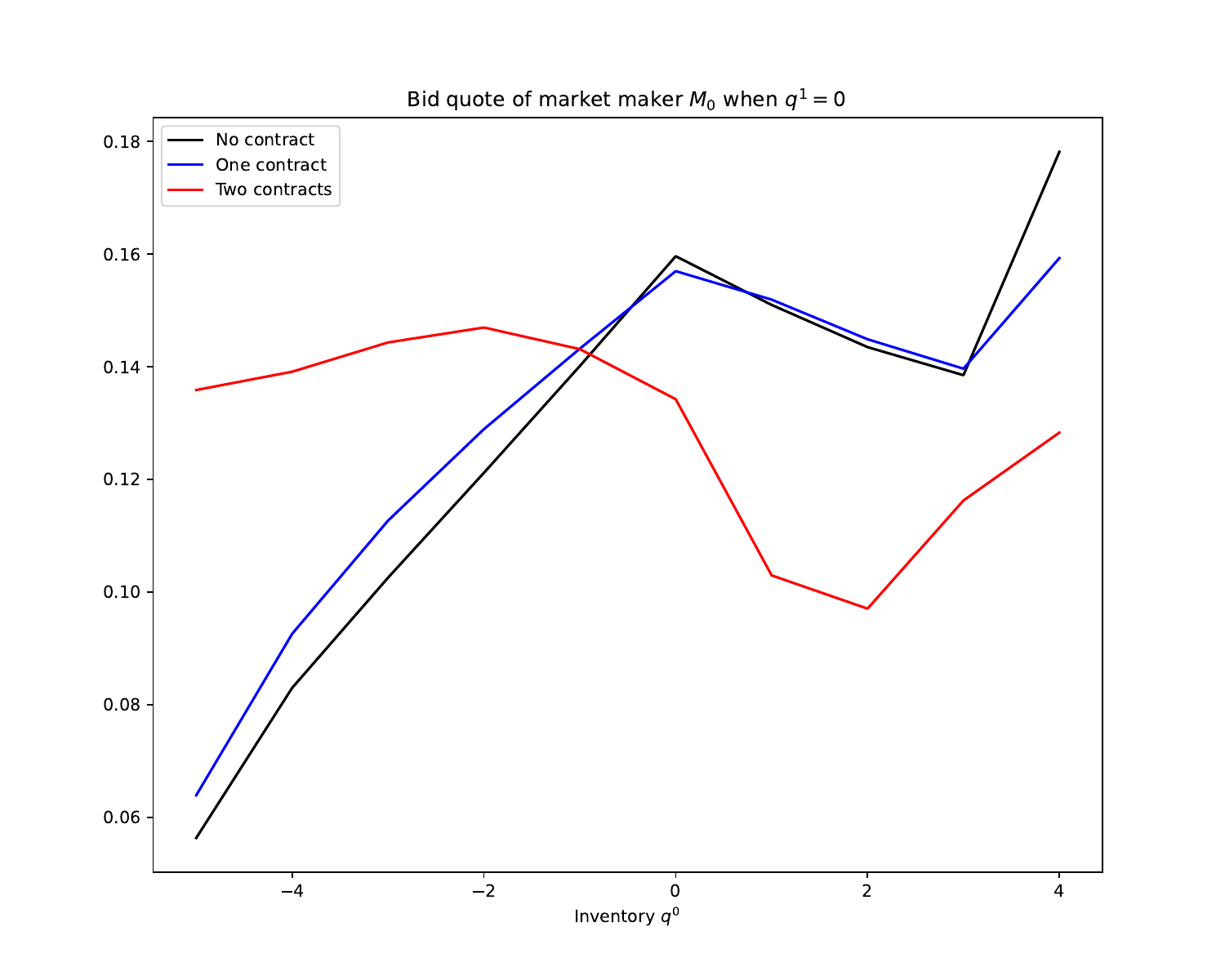}
\end{subfigure}\hfill 
\begin{subfigure}{.5\linewidth}
  \includegraphics[width=\linewidth]{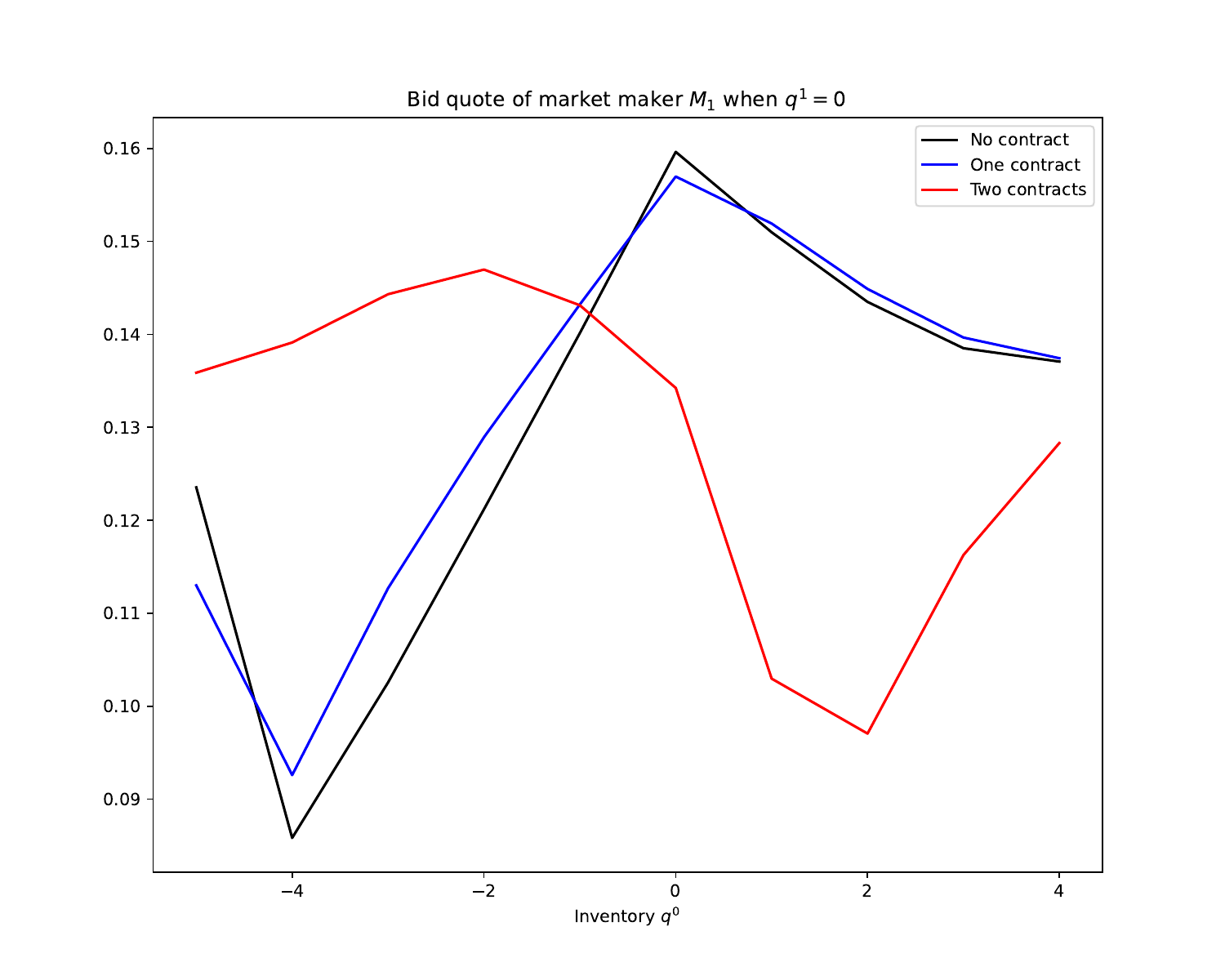}
\end{subfigure}
\caption{Optimal bid quote of the market makers as a function of the inventory of market maker $M_0$ on the left and $M_1$ on the right.}
\label{fig:deltas}
\end{figure}

\section{Conclusion}

This paper provides a model of interaction between exchanges sharing a limit order book. It shows that a passive exchange, i.e. an exchange who is not engaged in incentive provision to her own market market, does benefit from the incentive mechanisms set by her competitor. As a consequence, we clearly identify a spill over effect of competitiveness between platforms which reduces their own incentives to get engaged in costly optimal contracts with their market markets. This paper shows that incentive provision by platforms to their market makers is a public good. As a consequence, regulation should take this fact in consideration when enforcing shared order books. A consequence of such regulation is to deter exchanges from developing innovative solutions to increase market liquidity because her competitors will benefit from these innovations without having to pay to cost of their implementation. 

\appendix

\section{Appendix}

\subsection{Proof of \Cref{FPMM}}

Notice that the function $h^0$ can be written as
$$h^0(d^0, d^1, z^0, q) = h^0_b(d^{b,0}, d^{b,1}, z^{0,b}, q) + h^0_a(d^{a,0}, d^{a,1}, z^{0,a}, q),$$
where for $k\in \{b,a\}$,
\begin{align*}
h^0_k(d^{k,0}, d^{k,1}, z^{0,k}, q) = &\sum_{j=0}^1 \frac{1 - e^{-\gamma^0 \Big(z^{0,k,0,j} + {d^{k,0}} - \beta \big( d^{k,0} - \underline{d^k} \big) \Big)}}{\gamma^0}  \mathds{1}_{\{q^0 \phi(k)  < \bar q \}} \Lambda^{j}\big(d^{k,0}\big)\\
    & + \sum_{j=0}^1 \frac{1 - e^{-\gamma^0 z^{0,k,1,j} }}{\gamma^0}  \mathds{1}_{\{q^1 \phi(k)  < \bar q \}} \Lambda^{j}\big(d^{k,1}\big).    
\end{align*}
Similarly, we can write 
$$h^1(d^0, d^1, z^1, q) = h^1_b(d^{b,0}, d^{b,1}, z^{1,b}, q) + h^1_a(d^{a,0}, d^{a,1}, z^{1,a}, q).$$
We prove the result for the bid side (the proof for the ask side is of course identical).\\

\begin{enumerate}[wide, labelindent=0pt]
    \item First, let us notice that if $q^0 = \bar q$, then for a fixed $(d^{b,1}, z^{0,b})$, the function $h^0_b(., d^{b,1}, z^{0,b}, \bar q)$ is actually constant and therefore the maximum is reached at any point. In particular, it is reached at $\delta_{\infty}$. In what follows, we will only consider fixed points $\delta^\star$ such that $\delta^{\star, k, i} (z,q) = \delta_\infty$ if $q^i\phi(k) = \bar q$.\\
    
    \item Let us now study the function $h^0_b$ when $q^0 < \bar q$.\\
    
    Assume first that $(d^{b,0}, d^{b,1})$ are such that $d^{b,0}<d^{b,1}$. Then we have
\begin{align*}
\frac{\partial h^0_b}{\partial d^{b,0}}   (d^{b,0}, d^{b,1}, z^{0,b}, q) &= \sum_{j=0}^1 \left(-\frac{A^j \kappa}{\sigma \gamma^0} e^{-\frac{\kappa}{\sigma}(d^{b,0}+c^j)} \left(1-e^{-\gamma^0 \left( z^{0,b,0,j} + d^{b,0}\right)} \right) + A^j e^{-\frac{\kappa}{\sigma}(d^{b,0}+c^j)} e^{-\gamma^0 \left( z^{0,b,0,j} + d^{b,0}\right)}\right)\\
&= e^{-\frac{\kappa}{\sigma}d^{b,0}} \left( e^{-\gamma^0 d^{b,0}} \left( \left( 1 + \frac{\kappa}{\sigma \gamma^0}\right) \sum_{j=0}^1 A^j e^{-\frac{\kappa}{\sigma} c^j - \gamma^0 z^{0,b,0,j}} \right) - \frac{\kappa}{\sigma \gamma^0} \sum_{j=0}^1 A^j e^{-\frac{\kappa}{\sigma} c^j} \right).
\end{align*}
The above derivative is equal to 0 when
$$d^{b,0} = \frac{1}{\gamma^0} \left( \log \bigg( 1 + \frac{\sigma \gamma^0}{\kappa}  \bigg) + \log \Bigg( \frac{\sum_{j=0}^1 A^j e^{-\frac{\kappa}{\sigma}c^j - \gamma^0 z^{0,k,0,j}}}{\sum_{j=0}^1 A^j e^{-\frac{\kappa}{\sigma}c^j}}  \Bigg)   \right).$$

From what precedes, for a fixed $(d^{b,1}, z^{0,b}, q)$, if the maximum of $h^0_b$ is reached on $[-\delta_{\infty}, d^{b,1}]$, then it should be equal to  $\Gamma^{b,0}(z^{0,b,0,:},q^0)$.\\

Assume now that $(d^{b,0}, d^{b,1})$ are such that $d^{b,0}>d^{b,1}$. Then we have
\begin{align*}
\frac{\partial h^0_b}{\partial d^{b,0}}   (d^{b,0},& d^{b,1}, z^{0,b}, q) = \sum_{j=0}^1 \Bigg(-\frac{A^j \kappa}{\sigma \gamma^0} e^{-\frac{\kappa}{\sigma}(d^{b,0}+c^j)} \left(1-e^{-\gamma^0 \left( z^{0,b,0,j} + (1-\beta)d^{b,0} + \beta d^{b,1}\right)} \right)\\
&\qquad \qquad \qquad \qquad + A^j(1-\beta) e^{-\frac{\kappa}{\sigma}(d^{b,0}+c^j)} e^{-\gamma^0 \left( z^{0,b,0,j} + (1-\beta)d^{b,0} + \beta d^{b,1}\right)}\Bigg)\\
&= e^{-\frac{\kappa}{\sigma}d^{b,0}} \left( e^{-\gamma^0 \left((1-\beta)d^{b,0} + \beta d^{b,1}\right)} \left( \left( 1 - \beta + \frac{\kappa}{\sigma \gamma^0}\right) \sum_{j=0}^1 A^j e^{-\frac{\kappa}{\sigma} c^j - \gamma^0 z^{0,b,0,j}} \right) - \frac{\kappa}{\sigma \gamma^0} \sum_{j=0}^1 A^j e^{-\frac{\kappa}{\sigma} c^j} \right).
\end{align*}
The above derivative is equal to 0 when
$$d^{b,0} = - \frac{\beta}{1-\beta} d^{b,1} + \frac{1}{\gamma^0 (1-\beta)} \left( \log \bigg( 1 + \frac{(1-\beta)\sigma \gamma^0}{\kappa}  \bigg) + \log \Bigg( \frac{\sum_{j=0}^1 A^j e^{-\frac{\kappa}{\sigma}c^j - \gamma^0 z^{0,k,0,j}}}{\sum_{j=0}^1 A^j e^{-\frac{k}{\sigma}c^j}}  \Bigg)   \right).$$

From what precedes, for a fixed $(d^{b,1}, z^{0,b}, q)$, if the maximum of $h^0_b$ is reached on $[d^{b,1}, +\delta_{\infty}]$, then it should be equal to $\Gamma^{b,0}_\beta (d^{b,1},z^{0,k,0,:},q^0)$.\\

Of course, the same results can be obtained concerning $h^1_b$.

\item A simple analysis of the different cases then allows to derive the associated equilibria.\\

If $\Gamma^{b,0}(z^{0,b,0,:},q^0)\le \Gamma^{b,1}(z^{1,b,1,:},q^{1}),\Gamma^{b,1}_\beta \left(\Gamma^{b,0}(z^{0,b,0,:},q^{0}),z^{1,b,1,:},q^1\right),$ then clearly the values
\begin{align*}
    d^{b,0} &= \Gamma^{b,0}(z^{0,b,0,:},q^0),\\
    d^{b,1} &= \Gamma^{b,1}_\beta \left(\Gamma^{b,0}(z^{0,b,0,:},q^{0}),z^{1,b,1,:},q^1\right)
\end{align*}
correspond to a fixed point.\\

Conversely, if $\Gamma^{b,1}(z^{1,b,1,:},q^1)\le \Gamma^{b,0}(z^{0,b,0,:},q^{0}),\Gamma^{b,0}_\beta \left(\Gamma^{b,1}(z^{1,b,1,:},q^1),z^{0,b,0,:},q^0\right),$ then the values
\begin{align*}
    d^{b,0} &= \Gamma^{b,0}_\beta \left(\Gamma^{b,1}(z^{1,b,1,:},q^1),z^{0,b,0,:},q^0\right),\\
    d^{b,1} &= \Gamma^{b,1}(z^{1,b,1,:},q^1)
\end{align*}
correspond to a fixed point.\\

However, if $\Gamma^{b,0}_\beta \left(\Gamma^{b,1}(z^{1,b,1,:},q^1),z^{0,b,0,:},q^0\right) \le \Gamma^{b,1}(z^{1,b,1,:},q^1) \le \Gamma^{b,0}(z^{0,b,0,:},q^{0}),$ then the fixed point is given by 
\begin{align*}
    d^{b,0} = d^{b,1} = \Gamma^{b,1}(z^{1,b,1,:},q^1),
\end{align*}
and conversely if $\Gamma^{b,1}_\beta \left(\Gamma^{b,0}(z^{0,b,0,:},q^{0}),z^{1,b,1,:},q^1\right) \le \Gamma^{b,0}(z^{0,b,0,:},q^0) \le \Gamma^{b,1}(z^{1,b,1,:},q^{1}),$ then the fixed point is given by 
\begin{align*}
    d^{b,0} = d^{b,1} = \Gamma^{b,0}(z^{0,b,0,:},q^0).
\end{align*}

Considering the same result on the ask side, this defines exactly the function $\Delta:\mathcal R^2 \times [-\bar q, \bar q]^2 \mapsto \mathcal M_2(\mathcal B_{\infty})$ such that for every $(z,q) \in \mathcal R^2 \times [-\bar q, \bar q]$ of \Cref{FPMM}, which is therefore in $\mathcal E$.

\item Notice that this equilibrium is unique up to the choice of $\delta^{\star, k, i} (z,q) = \delta_\infty$ if $q^i\phi(k) = \bar q$, which we imposed in our framework.\\

Hence, we finally have $\bar{\mathcal E}={\mathcal E} =\{\Delta\}$.

\end{enumerate}

\subsection{Dynamic programming principle}

For $i\in \{0,1\}$, a $\mathbb F-$predictable, $[0,T]-$valued stopping time $\tau$, an admissible couple of contracts $\xi \in \mathcal C$, and an admissible strategy $\tilde \delta = \left(\tilde \delta^0, \tilde \delta^1\right)$, we define
\begin{align*}
    J_{M_i} \left(\tau, \tilde \delta^0, \tilde \delta^1, \xi^i \right) = \mathbb E^{\tilde \delta}_\tau \left[ -e^{-\gamma^i \bigg(\int_\tau^T \Big(\tilde \delta^{a,i}_t - \beta \big(\tilde \delta^{a,i}_t - \underline{\tilde \delta^{a}_t}  \big) \Big) \mathrm dN^{a,i}_t + \int_\tau^T \Big(\tilde \delta^{b,i}_t - \beta \big( \tilde \delta^{b,i}_t - \underline{\tilde \delta^{b}_t} \big) \Big) \mathrm dN^{b,i}_t + \int_\tau^T Q^{i}_t \mathrm dS_t + \xi^i \bigg)} \right].
\end{align*}
We also define
$$\mathcal J_{\tau}^0 = \Bigg(J_{M_0} \bigg(\tau, \delta^0, \tilde \delta^1, \xi^0 \bigg) \Bigg)_{\delta^0 \in \mathcal A^0(\tilde \delta^1)}, \qquad V_{M_0}\left(\tau, \tilde \delta^1, \xi^0\right) = \underset{\delta^0 \in \mathcal A^0(\tilde \delta^1)}{\esssup} J_{M_0} \bigg(\tau, \delta^0, \tilde \delta^1, \xi^0 \bigg),$$
and
$$\mathcal J_{\tau}^1 = \Bigg(J_{M_1} \bigg(\tau, \tilde \delta^0, \delta^1, \xi^1 \bigg) \Bigg)_{\delta^1 \in \mathcal A^1(\tilde \delta^0)}, \qquad V_{M_1}\left(\tau, \tilde \delta^0, \xi^1\right) = \underset{\delta^1 \in \mathcal A^1(\tilde \delta^0)}{\esssup} J_{M_1} \bigg(\tau, \tilde \delta^0,  \delta^1, \xi^1 \bigg).$$

\begin{lemme}\label{seqdel}
    Let $t\in [0,T]$ and $\tau$ be a $\mathbb F-$predictable stopping time with value in $[t,T].$ Then there exists a non-decreasing sequence $(\delta^{0,n})_{n\in \mathbb N}$ in $\mathcal A^0(\tilde \delta^1)$ such that
    $$V_{M_0}\left(\tau, \tilde \delta^1, \xi^0\right) = \underset{n\rightarrow +\infty}{\lim} J_{M_0} \bigg(\tau, \delta^{0,n}, \tilde \delta^1, \xi^0 \bigg),$$
    and similarly there exists a non-decreasing sequence $(\delta^{1,n})_{n\in \mathbb N}$ in $\mathcal A^1(\tilde \delta^0)$ such that
    $$V_{M_1}\left(\tau, \tilde \delta^0, \xi^1\right) = \underset{n\rightarrow +\infty}{\lim} J_{M_1} \bigg(\tau, \tilde \delta^{0}, \delta^{1,n}, \xi^1 \bigg).$$
\end{lemme}

\begin{proof}
    We only prove the result for $V_{M_0}$, as the proof for $V_{M_1}$ is similar.\\

    For $\delta^0$ and $\delta^{0'}$ in $\mathcal A^0(\tilde \delta^1)$ define
    $$\hat \delta^0 = \delta^0 \mathds 1_{\{J_{M_0} (\tau, \delta^0, \tilde \delta^1, \xi^0 ) \ge J_{M_0} (\tau, \delta^{0'}, \tilde \delta^1, \xi^0 )\}} + \delta^{0'} \mathds 1_{\{J_{M_0} (\tau, \delta^0, \tilde \delta^1, \xi^0 ) < J_{M_0} (\tau, \delta^{0'}, \tilde \delta^1, \xi^0 )\}}.$$
    Then $\hat \delta^0 \in \mathcal A^0(\tilde \delta^1)$, and by definition of $\hat \delta^0$:
    $$J_{M_0} (\tau, \hat \delta^0, \tilde \delta^1, \xi^0 ) \ge \max \left(J_{M_0} (\tau, \delta^0, \tilde \delta^1, \xi^0 ), J_{M_0} (\tau, \delta^{0'}, \tilde \delta^1, \xi^0 ) \right).$$
    Hence $\mathcal J_{\tau}^0 $ is directed upwards, and the result follows from \cite[Proposition VI.I.I]{neveu1975discrete}.
\end{proof}

\begin{lemme}\label{PPD}
    Let $t\in [0,T]$ and $\tau$ be a $\mathbb F-$predictable stopping time with value in $[t,T].$ Then
    \begin{align*}
       &V_{M_0}\left(t, \tilde \delta^1, \xi^0\right)\\
       &= \underset{\delta^0 \in \mathcal A^0(\tilde \delta^1)}{\esssup} \mathbb E^{\delta^0 \otimes \tilde \delta^1}_t \Bigg[ e^{-\gamma^0 \bigg(\int_t^\tau \Big( \delta^{a,0}_s - \beta \big( \delta^{a,0}_s - \underline{\delta^{a,0}_s \otimes \tilde \delta^{a,1}_s}  \big) \Big) \mathrm dN^{a,0}_s + \int_t^\tau \Big( \delta^{b,0}_s - \beta \big(  \delta^{b,0}_s - \underline{\delta^{b,0}_s \otimes \tilde \delta^{b,1}_s} \big) \Big) \mathrm dN^{b,0}_s + \int_t^\tau Q^{0}_s \mathrm dS_s\bigg)}\\
       & \qquad \qquad \qquad \qquad \qquad \qquad \qquad \qquad \times V_{M_0}\left(\tau, \tilde \delta^1, \xi^0\right) \Bigg], 
    \end{align*}
    and similarly
    \begin{align*}
       &V_{M_1}\left(t, \tilde \delta^0, \xi^1\right)\\
       &= \underset{\delta^1 \in \mathcal A^1(\delta^0)}{\esssup} \mathbb E^{\tilde \delta^0 \otimes \delta^1}_t \Bigg[ e^{-\gamma^1 \bigg(\int_t^\tau \Big( \delta^{a,1}_s - \beta \big(\delta^{a,1}_s - \underline{\tilde \delta^{a,0}_s \otimes  \delta^{a,1}_s}  \big) \Big) \mathrm dN^{a,1}_s + \int_t^\tau \Big(\delta^{b,1}_s - \beta \big(  \delta^{b,1}_s - \underline{\tilde \delta^{b,0}_s \otimes  \delta^{b,1}_s} \big) \Big) \mathrm dN^{b,1}_s + \int_t^\tau Q^{1}_s \mathrm dS_s\bigg)}\\
       & \qquad \qquad \qquad \qquad \qquad \qquad \qquad \qquad \times  V_{M_1}\left(\tau, \tilde \delta^0, \xi^1\right) \Bigg].
    \end{align*}
\end{lemme}

\begin{proof}
    We only prove the result for $V_{M_0}$, as the proof for $V_{M_1}$ is similar.\\
    
    Let $t\in [0,T]$ and $\tau$ be a $\mathbb F-$predictable stopping time with value in $[t,T].$ First, we have
    \begin{align*}
        &V_{M_0}\left(t, \tilde \delta^1, \xi^0\right)\\
        & = \underset{\delta^0 \in \mathcal A^0(\tilde \delta^1)}{\esssup} \mathbb E^{\delta^0 \otimes \tilde \delta^1}_t \left[ -e^{-\gamma^0 \bigg(\int_t^\tau \Big( \delta^{a,0}_s - \beta \big( \delta^{a,0}_s - \underline{\delta^{a,0}_s \otimes \tilde \delta^{a,1}_s}  \big) \Big) \mathrm dN^{a,0}_s + \int_t^\tau \Big( \delta^{b,0}_s - \beta \big(  \delta^{b,0}_s - \underline{\delta^{b,0}_s \otimes \tilde \delta^{b,1}_s} \big) \Big) \mathrm dN^{b,0}_s + \int_t^\tau Q^{0}_s \mathrm dS_s + \xi^0 \bigg)} \right]\\
        & = \underset{\delta^0 \in \mathcal A^0(\tilde \delta^1)}{\esssup} \mathbb E^{\delta^0 \otimes \tilde \delta^1}_t \left[ e^{-\gamma^0 \bigg(\int_t^\tau \Big( \delta^{a,0}_s - \beta \big( \delta^{a,0}_s - \underline{\delta^{a,0}_s \otimes \tilde \delta^{a,1}_s}  \big) \Big) \mathrm dN^{a,0}_s + \int_t^\tau \Big(\delta^{b,0}_s - \beta \big(  \delta^{b,0}_s - \underline{\delta^{b,0}_s \otimes \tilde \delta^{b,1}_s} \big) \Big) \mathrm dN^{b,0}_s + \int_t^\tau Q^{0}_s \mathrm dS_s\bigg)} \right.\\
        &\qquad  \left .\times \mathbb E^{\delta^0 \otimes \tilde \delta^1}_\tau \left[-e^{-\gamma^0 \bigg(\int_\tau^T \Big( \delta^{a,0}_s - \beta \big( \delta^{a,0}_s - \underline{\delta^{a,0}_s \otimes \tilde \delta^{a,1}_s}  \big) \Big) \mathrm dN^{a,0}_s + \int_\tau^T \Big( \delta^{b,0}_s - \beta \big( \delta^{b,0}_s - \underline{\delta^{b,0}_s \otimes \tilde \delta^{b,1}_s} \big) \Big) \mathrm dN^{b,0}_s + \int_\tau^T Q^{0}_s \mathrm dS_s + \xi^0 \bigg)}\right] \right]\\
        & \le\!\!\! \underset{\delta^0 \in \mathcal A^0(\tilde \delta^1)}{\esssup}\!\! \mathbb E^{\delta^0 \otimes \tilde \delta^1}_t\!\!\! \left[ e^{-\gamma^0\! \bigg(\!\!\int_t^\tau\! \Big( \delta^{a,0}_s - \beta \big( \delta^{a,0}_s - \underline{\delta^{a,0}_s \otimes \tilde \delta^{a,1}_s}  \big) \Big) \mathrm dN^{a,0}_s + \int_t^\tau\! \Big( \delta^{b,0}_s - \beta \big(\delta^{b,0}_s - \underline{\delta^{b,0}_s \otimes \tilde \delta^{b,1}_s} \big) \Big) \mathrm dN^{b,0}_s + \int_t^\tau\! Q^{0}_s \mathrm dS_s\!\!\bigg)} V_{M_0}\left(\tau, \tilde \delta^1, \xi^0\right)  \right]\!\!.
    \end{align*}

    We now prove the other inequality. Let $\delta^0, {\delta^0}' \in \mathcal A^0(\tilde \delta^1)$ and define
    $$\delta^{0,\tau}_s = \delta^0_s \mathds{1}_{\{0\le s <\tau\}} +  {\delta^0}'_s\mathds{1}_{\{ \tau \le s \le T\}}.$$ Then $\delta^{0,\tau} \in \mathcal A^0(\tilde \delta^1)$ and
    \begin{align*}
        &V_{M_0}\left(t, \tilde \delta^1, \xi^0\right)\\
        &\ge  \mathbb E^{\delta^{0,\tau} \otimes \tilde \delta^1}_t \left[ -e^{-\gamma^0 \bigg(\int_t^\tau \Big(\delta^{a,0, \tau}_s - \beta \big( \delta^{a,0, \tau}_s - \underline{\delta^{a,0,\tau}_s \otimes \tilde \delta^{a,1}_s}  \big) \Big) \mathrm dN^{a,0}_s + \int_t^\tau \Big(\delta^{b,0, \tau}_s - \beta \big( \delta^{b,0, \tau}_s - \underline{\delta^{b,0,\tau}_s \otimes \tilde \delta^{b,1}_s} \big) \Big) \mathrm dN^{b,0}_s + \int_t^\tau Q^{0}_s \mathrm dS_s + \xi^0 \bigg)} \right]\\
        &= \mathbb E^{\delta^{0,\tau} \otimes \tilde \delta^1}_t \left[ e^{-\gamma^0 \bigg(\int_t^\tau \Big(\delta^{a,0, \tau}_s - \beta \big(\delta^{a,0, \tau}_s - \underline{\delta^{a,0,\tau}_s \otimes \tilde \delta^{a,1}_s}  \big) \Big) \mathrm dN^{a,0}_s + \int_t^\tau \Big(\delta^{b,0, \tau}_s - \beta \big( \delta^{b,0, \tau}_s - \underline{\delta^{b,0,\tau}_s \otimes \tilde \delta^{b,1}_s} \big) \Big) \mathrm dN^{b,0}_s + \int_t^\tau Q^{0}_s \mathrm dS_s\bigg)} \right.\\
        &\qquad  \left .\times \mathbb E^{\delta^{0,\tau} \otimes \tilde \delta^1}_\tau \!\!\! \left[-e^{-\gamma^0 \bigg(\!\int_\tau^T \!\Big(\delta^{a,0, \tau}_s - \beta \big(\delta^{a,0, \tau}_s - \underline{\delta^{a,0,\tau}_s \otimes \tilde \delta^{a,1}_s}  \big) \Big) \mathrm dN^{a,0}_s + \int_\tau^T\! \Big(\delta^{b,0, \tau}_s - \beta \big( \delta^{b,0, \tau}_s - \underline{\delta^{b,0,\tau}_s \otimes \tilde \delta^{b,1}_s} \big) \Big) \mathrm dN^{b,0}_s + \int_\tau^T\! Q^{0}_s \mathrm dS_s + \xi^0 \!\bigg)}\right] \right]\!\!.\\
    \end{align*}

    Notice now that $\frac{L_T^{\delta^{0,\tau}\otimes \tilde \delta^1}}{L_\tau^{\delta^{0,\tau}\otimes \tilde \delta^1}} = \frac{L_T^{{\delta^{0}}'\otimes \tilde \delta^1}}{L_\tau^{{\delta^{0}}'\otimes \tilde \delta^1}}$, and therefore:
    \begin{align*}
      & \mathbb E^{\delta^{0,\tau} \otimes \tilde \delta^1}_\tau \left[-e^{-\gamma^0 \bigg(\int_\tau^T \Big(\delta^{a,0, \tau}_s - \beta \big(\delta^{a,0, \tau}_s - \underline{\delta^{a,0,\tau}_s \otimes \tilde \delta^{a,1}_s}  \big) \Big) \mathrm dN^{a,0}_s + \int_\tau^T \Big(\delta^{b,0, \tau}_s - \beta \big( \delta^{b,0, \tau}_s - \underline{\delta^{b,0,\tau}_s \otimes \tilde \delta^{b,1}_s} \big) \Big) \mathrm dN^{b,0}_s + \int_\tau^T Q^{0}_s \mathrm dS_s + \xi^0 \bigg)}\right] \\
      = & \mathbb E^{\mathbb P}_\tau \left[-e^{-\gamma^0 \bigg(\int_\tau^T \Big(\delta^{a,0, \tau}_s - \beta \big(\delta^{a,0, \tau}_s - \underline{\delta^{a,0,\tau}_s \otimes \tilde \delta^{a,1}_s}  \big) \Big) \mathrm dN^{a,0}_s + \int_\tau^T \Big(\delta^{b,0, \tau}_s - \beta \big( \delta^{b,0, \tau}_s - \underline{\delta^{b,0,\tau}_s \otimes \tilde \delta^{b,1}_s} \big) \Big) \mathrm dN^{b,0}_s + \int_\tau^T Q^{0}_s \mathrm dS_s + \xi^0 \bigg)} \frac{L_T^{\delta^{0,\tau}\otimes \tilde \delta^1}}{L_\tau^{\delta^{0,\tau}\otimes \tilde \delta^1}}\right]\\
      = & J_{M_0} \left(\tau, {\delta^{0}}', \tilde \delta^1, \xi^0 \right).
    \end{align*}

    Hence we obtain
    \begin{align*}
        &V_{M_0}\left(t, \tilde \delta^1, \xi^0\right)\\
        \ge &\mathbb E^{\delta^{0,\tau} \otimes \tilde \delta^1}_t \Bigg[ e^{-\gamma^0 \bigg(\int_t^\tau \Big(\delta^{a,0, \tau}_s - \beta \big(\delta^{a,0, \tau}_s - \underline{\delta^{a,0,\tau}_s \otimes \tilde \delta^{a,1}_s}  \big) \Big) \mathrm dN^{a,0}_s + \int_t^\tau \Big(\delta^{b,0, \tau}_s - \beta \big( \delta^{b,0, \tau}_s - \underline{\delta^{b,0,\tau}_s \otimes \tilde \delta^{b,1}_s} \big) \Big) \mathrm dN^{b,0}_s + \int_t^\tau Q^{0}_s \mathrm dS_s\bigg)}\\
        & \qquad \qquad \times J_{M_0} \left(\tau, {\delta^{0}}', \tilde \delta^1, \xi^0 \right)\Bigg].
    \end{align*}

    Using now the fact that $\frac{L_\tau^{\delta^{0,\tau}\otimes \tilde \delta^1}}{L_t^{\delta^{0,\tau}\otimes \tilde \delta^1}} = \frac{L_\tau^{{\delta^{0}}\otimes \tilde \delta^1}}{L_t^{{\delta^{0}}\otimes \tilde \delta^1}}$, we get:
    \begin{align*}
        &V_{M_0}\left(t, \tilde \delta^1, \xi^0\right)\\
        \ge &\mathbb E^{\mathbb P}_t \Bigg[ \frac{L_T^{\delta^{0,\tau}\otimes \tilde \delta^1}}{L_t^{\delta^{0,\tau}\otimes \tilde \delta^1}} e^{-\gamma^0 \bigg(\int_t^\tau \Big(\delta^{a,0, \tau}_s - \beta \big(\delta^{a,0, \tau}_s - \underline{\delta^{a,0,\tau}_s \otimes \tilde \delta^{a,1}_s} \big) \Big) \mathrm dN^{a,0}_s + \int_t^\tau \Big(\delta^{b,0, \tau}_s - \beta \big( \delta^{b,0, \tau}_s - \underline{\delta^{b,0,\tau}_s \otimes \tilde \delta^{b,1}_s} \big) \Big) \mathrm dN^{b,0}_s + \int_t^\tau Q^{0}_s \mathrm dS_s\bigg)}\\
        &\qquad \qquad \times J_{M_0} \left(\tau, {\delta^{0}}', \tilde \delta^1, \xi^0 \right)\Bigg]\\
        =& \mathbb E^{\mathbb P}_t\!\! \Bigg[\mathbb E^{\mathbb P}_\tau \!\!\Bigg[ \frac{L_T^{\delta^{0,\tau}\otimes \tilde \delta^1}}{L_t^{\delta^{0,\tau}\otimes \tilde \delta^1}} e^{-\gamma^0 \bigg(\int_t^\tau \Big(\delta^{a,0, \tau}_s - \beta \big(\delta^{a,0, \tau}_s - \underline{\delta^{a,0,\tau}_s \otimes \tilde \delta^{a,1}_s}  \big) \Big) \mathrm dN^{a,0}_s + \int_t^\tau \Big(\delta^{b,0, \tau}_s - \beta \big( \delta^{b,0, \tau}_s - \underline{\delta^{b,0,\tau}_s \otimes \tilde \delta^{b,1}_s} \big) \Big) \mathrm dN^{b,0}_s + \int_t^\tau Q^{0}_s \mathrm dS_s\bigg)}\\
        & \qquad \qquad \times J_{M_0} \left(\tau, {\delta^{0}}', \tilde \delta^1, \xi^0 \right)\Bigg] \Bigg]\\
         =& \mathbb E^{\mathbb P}_t\!\! \left[\!\mathbb E^{\mathbb P}_\tau\!\! \left[ \frac{L_T^{\delta^{0,\tau}\otimes \tilde \delta^1}}{L_\tau^{\delta^{0,\tau}\otimes \tilde \delta^1}}\! \right]\!\! \frac{L_\tau^{\delta^{0,\tau}\otimes \tilde \delta^1}}{L_t^{\delta^{0,\tau}\otimes \tilde \delta^1}} e^{-\gamma^0\! \bigg(\!\!\int_t^\tau\! \Big(\delta^{a,0, \tau}_s\! - \beta \big(\delta^{a,0, \tau}_s - \underline{\delta^{a,0,\tau}_s \otimes \tilde \delta^{a,1}_s}  \big) \Big) \mathrm dN^{a,0}_s + \int_t^\tau\! \Big(\delta^{b,0, \tau}_s\! - \beta \big( \delta^{b,0, \tau}_s - \underline{\delta^{b,0,\tau}_s \otimes \tilde \delta^{b,1}_s} \big) \Big) \mathrm dN^{b,0}_s + \int_t^\tau\! Q^{0}_s \mathrm dS_s\!\!\bigg)}\right. \\
         &  \qquad \qquad \times \left. \vphantom{\mathbb E^{\mathbb P}_\tau \left[ \frac{L_T^{\delta^{0,\tau}\otimes \tilde \delta^1}}{L_\tau^{\delta^{0,\tau}\otimes \tilde \delta^1}} \right] \frac{L_\tau^{\delta^{0,\tau}\otimes \tilde \delta^1}}{L_t^{\delta^{0,\tau}\otimes \tilde \delta^1}} e^{-\gamma^0 \bigg(\int_t^\tau \Big(\delta^{a,0, \tau}_s - \beta \big(\delta^{a,0, \tau}_s - \underline{\tilde \delta^{a}_s}  \big) \Big) \mathrm dN^{a,0}_s + \int_t^\tau \Big(\delta^{b,0, \tau}_s - \beta \big( \delta^{b,0, \tau}_s - \underline{\tilde \delta^{b}_s} \big) \Big) \mathrm dN^{b,0}_s + \int_t^\tau Q^{0}_s \mathrm dS_s\bigg)}} J_{M_0} \left(\tau, {\delta^{0}}', \tilde \delta^1, \xi^0 \right)\right] \\
         =& \mathbb E^{\mathbb P}_t \Bigg[ \frac{L_\tau^{\delta^{0,\tau}\otimes \tilde \delta^1}}{L_t^{\delta^{0,\tau}\otimes \tilde \delta^1}} e^{-\gamma^0 \bigg(\int_t^\tau \Big(\delta^{a,0, \tau}_s - \beta \big(\delta^{a,0, \tau}_s - \underline{\delta^{a,0,\tau}_s \otimes \tilde \delta^{a,1}_s}  \big) \Big) \mathrm dN^{a,0}_s + \int_t^\tau \Big(\delta^{b,0, \tau}_s - \beta \big( \delta^{b,0, \tau}_s - \underline{\delta^{b,0,\tau}_s \otimes \tilde \delta^{b,1}_s} \big) \Big) \mathrm dN^{b,0}_s + \int_t^\tau Q^{0}_s \mathrm dS_s\bigg)}\\
         &\qquad \qquad \times J_{M_0} \left(\tau, {\delta^{0}}', \tilde \delta^1, \xi^0 \right)\Bigg]\\
         =& \mathbb E^{\delta^{0} \otimes \tilde \delta^1}_t \Bigg[e^{-\gamma^0 \bigg(\int_t^\tau \Big(\delta^{a,0, \tau}_s - \beta \big(\delta^{a,0, \tau}_s - \underline{\delta^{a,0,\tau}_s \otimes \tilde \delta^{a,1}_s}  \big) \Big) \mathrm dN^{a,0}_s + \int_t^\tau \Big(\delta^{b,0, \tau}_s - \beta \big( \delta^{b,0, \tau}_s - \underline{\delta^{b,0,\tau}_s \otimes \tilde \delta^{b,1}_s} \big) \Big) \mathrm dN^{b,0}_s + \int_t^\tau Q^{0}_s \mathrm dS_s\bigg)}\\
         &\qquad \qquad \times J_{M_0} \left(\tau, {\delta^{0}}', \tilde \delta^1, \xi^0 \right)\Bigg].
    \end{align*}

    This inequality holds for all ${\delta^0}' \in \mathcal A^0(\tilde \delta^1)$, hence using \Cref{seqdel} and monotone convergence, we get:
    \begin{align*}
        &V_{M_0}\left(t, \tilde \delta^1, \xi^0\right)\\
        \ge & \mathbb E^{\delta^{0} \otimes \tilde \delta^1}_t \Bigg[e^{-\gamma^0 \bigg(\int_t^\tau \Big(\delta^{a,0, \tau}_s - \beta \big(\delta^{a,0, \tau}_s - \underline{\delta^{a,0,\tau}_s \otimes \tilde \delta^{a,1}_s}  \big) \Big) \mathrm dN^{a,0}_s + \int_t^\tau \Big(\delta^{b,0, \tau}_s - \beta \big( \delta^{b,0, \tau}_s - \underline{\delta^{b,0,\tau}_s \otimes \tilde \delta^{b,1}_s} \big) \Big) \mathrm dN^{b,0}_s + \int_t^\tau Q^{0}_s \mathrm dS_s\bigg)}\\
        &\qquad \qquad \times V_{M_0} \left(\tau, \tilde \delta^1, \xi^0 \right)\Bigg],
    \end{align*}
    which proves the result.

\end{proof}

\subsection{Proof of \Cref{contractrep} and \Cref{NashMM}}

Let us start with the following lemma.

\begin{lemme}\label{leminteg}
    For all $\delta \in \mathcal A$, the processes $V_{M_0}\left(., \delta^1, \xi^0\right)$ and $V_{M_1}\left(., \delta^0, \xi^1\right)$ are negative, and there exists $p>0$ such that
    $$\mathbb E^{\delta} \left[ \underset{t \in [0,T]}{\sup} \left| V_{M_0}\left(t, \delta^1, \xi^0\right)\right|^{1+p}\right]<+\infty,$$ 
    $$\mathbb E^{\delta} \left[ \underset{(u,t) \in [0,T]^2}{\sup} \left( e^{-\gamma^0 \bigg(\int_u^t \Big( \delta^{a,0}_s - \beta \big( \delta^{a,0}_s - \underline{ \delta^{a}_s}  \big) \Big) \mathrm dN^{a,0}_s + \int_u^t \Big( \delta^{b,0}_s - \beta \big(  \delta^{b,0}_s - \underline{ \delta^{b}_s} \big) \Big) \mathrm dN^{b,0}_s + \int_u^t Q^{0}_s \mathrm dS_s\bigg)}\right)^{1+p}\right]<+\infty$$
    and
    $$\mathbb E^{\delta} \left[ \underset{t \in [0,T]}{\sup} \left| V_{M_1}\left(t, \delta^0, \xi^1\right)\right|^{1+p}\right]<+\infty,$$ 
    $$\mathbb E^{\delta} \left[ \underset{(u,t) \in [0,T]^2}{\sup} \left( e^{-\gamma^1 \bigg(\int_u^t \Big( \delta^{a,1}_s - \beta \big( \delta^{a,1}_s - \underline{ \delta^{a}_s}  \big) \Big) \mathrm dN^{a,1}_s + \int_u^t \Big( \delta^{b,1}_s - \beta \big(  \delta^{b,1}_s - \underline{ \delta^{b}_s} \big) \Big) \mathrm dN^{b,1}_s + \int_u^t Q^{1}_s \mathrm dS_s\bigg)}\right)^{1+p}\right]<+\infty.$$
\end{lemme}

\begin{proof}

We only prove the result for ${M_0}$, as the proof for ${M_1}$ is similar.\\

Let us first show that $V_{M_0}\left(., \delta^1, \xi^0\right)$ is a negative process. We have:
\begin{align*}
    &\frac{L^{\delta}_T}{L^{\delta}_t}\\
    \ge & \alpha_{t,T} := \exp \Bigg( \sum_{(i,j) \in \{0,1\}} \sum_{k\in \{b,a\}} \bigg( \Big( \log (A^j) - \frac{\kappa}{\sigma}(\delta_{\infty} + c^j)\Big) (N^{k,i,j}_T-N^{k,i,j}_t)  - \Big( A^j e^{-\frac{\kappa}{\sigma}(c-\delta_\infty)}-1\Big) (T-t)  \bigg)  \Bigg)\\
    > & 0.
\end{align*}
Hence we have
\begin{align*}
    V_{M_0}\left(t, \delta^1, \xi^0\right) \le \mathbb E^{\mathbb P} \left[-\alpha_{t,T}e^{-\gamma^0 \bigg(\delta_{\infty} (N^{a,0}_T - N^{a,0}_t + N^{b,0}_T - N^{b ,0}_t)  \int_t^T Q^{0}_s \mathrm dS_s + \xi^0 \bigg)} \right] <0.
\end{align*}

Furthermore, we have
\begin{align*}
    V_{M_0}\left(t, \delta^1, \xi^0\right) = &\underset{\delta^0 \in \mathcal A^0(\delta^1)}{\esssup} \mathbb E^{\delta^0 \otimes  \delta^1}_t \left[ -e^{-\gamma^0 \bigg(\int_t^T \Big( \delta^{a,0}_s - \beta \big( \delta^{a,0}_s - \underline{\delta^{a}_s}  \big) \Big) \mathrm dN^{a,0}_s + \int_t^T \Big( \delta^{b,0}_s - \beta \big(  \delta^{b,0}_s - \underline{ \delta^{b}_s} \big) \Big) \mathrm dN^{b,0}_s + \int_t^T Q^{0}_s \mathrm dS_s + \xi^0 \bigg)} \right]\\
    \ge &\   \mathbb E^{\delta^0 \otimes  \delta^1}_t \left[ -e^{\gamma^0 \bigg( \delta_{\infty}(N^{a,0}_T + N^{b,0}_T) - \int_t^T Q^{0}_s \mathrm dS_s - \xi^0 \bigg)} \right]\ge    \mathbb E^{\delta^0 \otimes  \delta^1}_t \left[ -e^{\gamma^0 \bigg( \delta_{\infty}(N^{a,0}_T + N^{b,0}_T) - \xi^0 \bigg)} \right]e^{ \frac { (\gamma^0)^2 \bar q^2 \sigma^2 T}{2}}
\end{align*}

with $\delta^0 \in \mathcal A^0(\delta^1)$. Therefore, writing $\delta =\delta^0 \otimes  \delta^1  $, we have for $p>0$
\begin{align*}
    \mathbb E^{\delta} \left[ \underset{t \in [0,T]}{\sup} \left| V_{M_0}\left(t, \delta^1, \xi^0\right)\right|^{1+p}\right] \le e^{ \frac {(1+p) (\gamma^0)^2 \bar q^2 \sigma^2 T}{2}}\mathbb E^{\delta} \left[ \underset{t \in [0,T]}{\sup}\mathbb E^{\delta^0 \otimes  \delta^1}_t \left[ e^{\gamma^0 \bigg( \delta_{\infty}(N^{a,0}_T + N^{b,0}_T) - \xi^0 \bigg)} \right]^{1+p} \right]. 
\end{align*}

By Hölder's inequality along with condition \eqref{integcond0}, the term inside the conditional expectation is integrable. Therefore, Doob's inequality gives us
$$\mathbb E^{\delta} \left[ \underset{t \in [0,T]}{\sup} \left| V_{M_0}\left(t, \delta^1, \xi^0\right)\right|^{1+p}\right] \le  C_p e^{ \frac {(1+p) (\gamma^0)^2 \bar q^2 \sigma^2 T}{2}}\mathbb E^{\delta} \left[  e^{(1+p)\gamma^0 \bigg( \delta_{\infty}(N^{a,0}_T + N^{b,0}_T) - \xi^0 \bigg)}  \right] < +\infty, $$
for a given $C_p >0$.\\

We can prove
$$\mathbb E^{\delta} \left[ \underset{(u,t) \in [0,T]^2}{\sup} \left( e^{-\gamma^0 \bigg(\int_u^t \Big( \delta^{a,0}_s - \beta \big( \delta^{a,0}_s - \underline{ \delta^{a}_s}  \big) \Big) \mathrm dN^{a,0}_s + \int_u^t \Big( \delta^{b,0}_s - \beta \big(  \delta^{b,0}_s - \underline{ \delta^{b}_s} \big) \Big) \mathrm dN^{b,0}_s + \int_u^t Q^{0}_s \mathrm dS_s\bigg)}\right)^{1+p}\right]<+\infty$$
using the same arguments.
    
\end{proof}

We now proceed in several steps to prove \Cref{contractrep} and \Cref{NashMM}.\\

\begin{enumerate}[wide, labelindent=0pt]
    \item Let $\xi=(\xi^0, \xi^1)\in \mathcal C$ and $\bar \delta \in \overline{NE}(\xi)$. It follows from the dynamic programming principle of \Cref{PPD} that, for $\delta^0 \in \mathcal A^0(\bar \delta^1)$, the process $\left(U^{0,\delta^0 \otimes \bar \delta^1}_t\right)_{t\in [0,T]}$ given by
    $$U^{0,\delta^0 \otimes \bar \delta^1}_t = V_{M_0}\left(t, \bar \delta^1, \xi^0\right) e^{-\gamma^0 \bigg(\int_0^t \Big( \delta^{a,0}_s - \beta \big( \delta^{a,0}_s - \underline{\delta^{a,0}_s \otimes \bar \delta^{a,1}_s}  \big) \Big) \mathrm dN^{a,0}_s + \int_0^t \Big( \delta^{b,0}_s - \beta \big(  \delta^{b,0}_s - \underline{\delta^{b,0}_s \otimes \bar \delta^{b,1}_s} \big) \Big) \mathrm dN^{b,0}_s + \int_0^t Q^{0}_s \mathrm dS_s\bigg)}$$
    is a $\mathbb P^{\delta^0 \otimes \bar \delta^1}-$supermartingale.\\

    For any $\mathbb F-$predictable stopping time $\tau$ with values in $[0,T]$, we have
    \begin{align*}
        V_{M_0}\left(\bar \delta^1, \xi^0\right) \ge& \mathbb E^{\bar \delta} \left[V_{M_0}\left(\tau, \bar \delta^1, \xi^0\right) e^{-\gamma^0 \bigg(\int_0^\tau \Big( \bar \delta^{a,0}_s - \beta \big( \bar \delta^{a,0}_s - \underline{ \bar \delta^{a}_s}  \big) \Big) \mathrm dN^{a,0}_s + \int_0^\tau \Big( \bar\delta^{b,0}_s - \beta \big(  \bar \delta^{b,0}_s - \underline{ \bar \delta^{b}_s} \big) \Big) \mathrm dN^{b,0}_s + \int_0^\tau Q^{0}_s \mathrm dS_s\bigg)}\right]\\
        \ge&  \mathbb E^{\bar \delta} \left[- e^{-\gamma^0 \bigg(\int_0^T \Big( \bar \delta^{a,0}_s - \beta \big( \bar \delta^{a,0}_s - \underline{ \bar \delta^{a}_s}  \big) \Big) \mathrm dN^{a,0}_s + \int_0^T \Big( \bar\delta^{b,0}_s - \beta \big(  \bar \delta^{b,0}_s - \underline{ \bar \delta^{b}_s} \big) \Big) \mathrm dN^{b,0}_s + \int_0^T Q^{0}_s \mathrm dS_s + \xi^i\bigg)}\right] = V_{M_0}\left(\bar \delta^1, \xi^0\right).
    \end{align*}
    Therefore, as the filtration is right-continuous, we deduce from \Cref{leminteg} that $\left(U^{0,\ \bar \delta}_t\right)_{t\in [0,T]}$ is a uniformly integrable $\mathbb P^{\bar \delta}-$martingale.\\

    As all probability measures are equivalent here, we deduce that $\left(U^{0,\delta^0 \otimes \bar \delta^1}_t\right)_{t\in [0,T]}$ admits a \textit{càdlàg} modification for any $\delta^0 \in \mathcal A^0(\bar \delta^1)$, and we can consider its Doob-Meyer decomposition given by:
    $$U^{0,\delta^0 \otimes \bar \delta^1}_t = M^{0,\delta^0 \otimes \bar \delta^1}_t - A^{0,\delta^0 \otimes \bar \delta^1, c}_t - A^{0,\delta^0 \otimes \bar \delta^1, d}_t \quad \forall t \in [0,T],$$
    where $\left(M^{0,\delta^0 \otimes \bar \delta^1}_t\right)_{t\in [0,T]}$ is a uniformly integrable $\mathbb P^{\delta^0 \otimes \bar \delta^1}-$martingale and
    $$A^{0,\delta^0 \otimes \bar \delta^1}_t  := A^{0,\delta^0 \otimes \bar \delta^1, c}_t + A^{0,\delta^0 \otimes \bar \delta^1, d}_t \quad \forall t \in [0,T]$$
    is a predictable, integrable, and non-decreasing process such that $A^{0,\delta^0 \otimes \bar \delta^1, c}_0 = A^{0,\delta^0 \otimes \bar \delta^1, d}_0 = 0$, with pathwise continuous component $\left(A^{0,\delta^0 \otimes \bar \delta^1, c}_t \right)_{t\in [0,T]}$ and a piecewise constant predictable process $\left(A^{0,\delta^0 \otimes \bar \delta^1, d}_t \right)_{t\in [0,T]}$.\\

    By the martingale representation theorem in \cite{el2018optimal}, we can write
    $$M^{0,\delta^0 \otimes \bar \delta^1}_t\!\!\! = V_{M_0}\left(\bar \delta^1, \xi^0\right) + \int_0^t \tilde Z^{0,S, \delta^0 \otimes \bar \delta^1}_s \mathrm dS_s + \sum_{i,j=0}^1 \int_0^t \tilde Z^{0, b, i, j, \delta^0 \otimes \bar \delta^1}_s \mathrm d\tilde N^{b,i,j, \delta^0 \otimes \bar \delta^1}_s + \sum_{i,j=0}^1 \int_0^t \tilde Z^{0,a, i, j,\delta^0 \otimes \bar \delta^1}_s \mathrm d\tilde N^{a,i,j, \delta^0 \otimes \bar \delta^1}_s$$
    where $\forall (i,j) \in \{0,1\}^2$, $\tilde N^{b,i,j, \delta^0 \otimes \bar \delta^1}$ and $\tilde N^{a,i,j, \delta^0 \otimes \bar \delta^1}$ are the compensated process associated to $N^{b,i,j}$ and $N^{a,i,j}$ under $\mathbb P^{\delta^0 \otimes \bar \delta^1}$, and $\tilde Z^{0,S, \delta^0 \otimes \bar \delta^1}$, $\tilde Z^{0, b, i, j, \delta^0 \otimes \bar \delta^1}$, and  $\tilde Z^{0, a, i, j, \delta^0 \otimes \bar \delta^1}$ are predictable processes.

    \item We now introduce the process $\left( Y_{M_0}\left(t, \bar \delta^1, \xi^0\right) \right)_{t \in [0,T]}$ given by
    $$Y_{M_0}\left(t, \bar \delta^1, \xi^0\right) = - \frac{1}{\gamma^0}\log \left(- V_{M_0}\left(t, \bar \delta^1, \xi^0\right) \right) \quad \forall t\in [0,T].$$

    As $A^{0,\delta^0 \otimes \bar \delta^1, d}$ is a predictable point process, and the jump times of  $N^{b,i,j}$ and $N^{a,i,j}$ are totally inaccessible stopping times under $\mathbb P$, it is clear that $\langle N^{b,i,j}, A^{0,\delta^0 \otimes \bar \delta^1, d}\rangle= \langle N^{a,i,j}, A^{0,\delta^0 \otimes \bar \delta^1, d}\rangle =0$ a.s. By Itô's formula, we get $Y_{M_0}\left(T, \bar \delta^1, \xi^0\right) = \xi^0$ and
    $$\mathrm dY_{M_0}\left(t, \bar \delta^1, \xi^0\right) = \sum_{i,j = 0}^1 Z^{0,b,i,j}_t \mathrm d N^{b,i,j}_t + \sum_{i,j = 0}^1 Z^{0,a,i,j}_t \mathrm d N^{a,i,j}_t + Z^{0,S}_t \mathrm dS_t - \mathrm d \tilde A^{0,d}_t - \mathrm d I^0_t,$$

    where we have by identification $\forall (i,j) \in \{0,1\}^2$ and $t \in [0,T]$
    \begin{align*}
        Z^{0,b,i,j}_t &= -\frac{1}{\gamma^0} \log \left( 1 + \frac{\tilde Z^{0, b, i, j, \delta^0 \otimes \bar \delta^1}_t}{U^{0,\delta^0 \otimes \bar \delta^1}_{t-}} \right) - \left(\delta^{b,0}_t - \beta \big(  \delta^{b,0}_t - \underline{ \delta^{b,0}_t \otimes \bar \delta^{b,1}_t} \big) \right),\\
        Z^{0,a,i,j}_t &= -\frac{1}{\gamma^0} \log \left( 1 + \frac{\tilde Z^{0, a, i, j, \delta^0 \otimes \bar \delta^1}_t}{U^{0,\delta^0 \otimes \bar \delta^1}_{t-}} \right) - \left(\delta^{a,0}_t - \beta \big(  \delta^{a,0}_t - \underline{ \delta^{a,0}_t \otimes \bar \delta^{a,1}_t} \big) \right),
    \end{align*}
    $$Z^{0,S}_t = - \frac{\tilde Z^{0,S, \delta^0 \otimes \bar \delta^1}_t}{\gamma^0 U^{0,\delta^0 \otimes \bar \delta^1}_{t-}} - Q^0_{t-}, \qquad \tilde A^{0,d}_t = \frac{1}{\gamma^0} \sum_{0\le s \le t} \log \left(1 - \frac{\Delta A^{0,\delta^0 \otimes \bar \delta^1, d}_t}{U^{0,\delta^0 \otimes \bar \delta^1}_{t-}}\right),$$
    and 
    $$I^0_t = \int _0^t \left( \bar h^0 \left(\delta^0_s,  \bar \delta^1_s, Z^0_s, Q_s\right) \mathrm ds - \frac 1{\gamma^0 U^{0,\delta^0 \otimes \bar \delta^1}_{s}} \mathrm d A^{0,\delta^0 \otimes \bar \delta^1, c}_s\right),$$
    with
    $$\bar h^0 \left(\delta^0_s,  \bar \delta^1_s, Z^0_s, Q_s\right) =  h^0 \left(\delta^0_s,  \bar \delta^1_s, Z^0_s, Q_s\right) - \frac 12 \gamma^0 \sigma^2 \left(Z^{0,S}_s \right)^2. $$
    This shows in particular that the process 
    $$\Delta a^0_t := - \frac{\Delta A^{0,\delta^0 \otimes \bar \delta^1, d}_t}{U^{0,\delta^0 \otimes \bar \delta^1}_{t-}} \ge 0$$
    is independent of $\delta^0 \in \mathcal A^0(\bar \delta^1)$.

    \item Notice now that
    \begin{align*}
        0 &= \underset{\delta^0 \in \mathcal A^0(\bar \delta^1)}{\sup} \mathbb E^{\delta^0 \otimes \bar \delta^1} \left[U^{0,\delta^0 \otimes \bar \delta^1}_{T} \right] - V_{M_0}\left(\bar \delta^1, \xi^0\right) = \underset{\delta^0 \in \mathcal A^0(\bar \delta^1)}{\sup} \mathbb E^{\delta^0 \otimes \bar \delta^1} \left[U^{0,\delta^0 \otimes \bar \delta^1}_{T} -  M^{0,\delta^0 \otimes \bar \delta^1}_{T} \right]\\
        &= \gamma^0 \underset{\delta^0 \in \mathcal A^0(\bar \delta^1)}{\sup} \mathbb E^{\mathbb P} \left[L^{\delta^0 \otimes \bar \delta^1}_{T} \int_0^T U^{0,\delta^0 \otimes \bar \delta^1}_{s-}  \left(\mathrm dI^0_s - \bar h^0 \left(\delta^0_s,  \bar \delta^1_s, Z^0_s, Q_s\right) \mathrm ds + \frac{\mathrm da^0_s}{\gamma^0} \right) \right].
    \end{align*}
    As the controls are uniformly bounded, we have by \Cref{leminteg} $\forall t \in [0,T]$ and $\forall \delta^0 \in \mathcal A^0(\bar \delta^1)$:
    $$U^{0,\delta^0 \otimes \bar \delta^1}_{t} \le -\beta^0_t = V_{M_0}\left(t, \bar \delta^1, \xi^0\right) e^{-\gamma^0\left( \delta_{\infty} \left(N^{a,0}_T- N^{a,0}_0 + N^{b,0}_T- N^{b,0}_0\right) + \int_0^t Q^0_s \mathrm dS_s\right)}<0.$$
    It follows that
    \begin{align*}
        0 &\le \underset{\delta^0 \in \mathcal A^0(\bar \delta^1)}{\sup} \mathbb E^{\mathbb P} \left[\alpha_{0,T} \int_0^T -\beta^0_{s-} \left(\mathrm dI^0_s - \bar h^0 \left(\delta^0_s,  \bar \delta^1_s, Z^0_s, Q_s\right) \mathrm ds + \frac{\mathrm da^0_s}{\gamma^0} \right) \right]\\
        &= - \mathbb E^{\mathbb P} \left[\alpha_{0,T} \int_0^T \beta^0_{s-} \left(\mathrm dI^0_s - \bar H^0 \left( \bar \delta^1_s, Z^0_s, Q_s\right) \mathrm ds + \frac{\mathrm da^0_s}{\gamma^0} \right) \right]
    \end{align*}
    where
    $$\bar H^0 \left(\delta^0_s,  \bar \delta^1_s, Z^0_s, Q_s\right) =  H^0 \left(  \bar \delta^1_s, Z^0_s, Q_s\right) - \frac 12 \gamma^0 \sigma^2 \left(Z^{0,S}_s \right)^2. $$

    Since the random variables $\alpha_{0,T} \int_0^T \beta^0_{s-} \left(\mathrm dI^0_s - \bar H^0 \left( \bar \delta^1_s, Z^0_s, Q_s\right) \mathrm ds  \right) $ and $\alpha_{0,T} \int_0^T \beta^0_{s-} \frac{\mathrm da^0_s}{\gamma^0}$ are non-negative, we finally obtain 
    $$\tilde A^{0,d} \qquad \text{and} \qquad I^0_t = \int_0^t \bar H^0 \left( \bar \delta^1_s, Z^0_s, Q_s\right) \mathrm ds.$$

    Moreover, by Itô's formula, 
    \begin{align*}
        \mathrm dU^{0,\delta^0 \otimes \bar \delta^1}_{t} &= \tilde Z^{0,S, \delta^0 \otimes \bar \delta^1}_t \mathrm dS_t + \sum_{i,j=0}^1 \tilde Z^{0, b, i, j, \delta^0 \otimes \bar \delta^1}_t \mathrm d\tilde N^{b,i,j, \delta^0 \otimes \bar \delta^1}_t + \sum_{i,j=0}^1 \tilde Z^{0,a, i, j,\delta^0 \otimes \bar \delta^1}_t \mathrm d\tilde N^{a,i,j, \delta^0 \otimes \bar \delta^1}_t\\
        & \qquad \qquad + \gamma^0 U^{0,\delta^0 \otimes \bar \delta^1}_{t} \Big( \bar H^0 \left( \bar \delta^1_t, Z^0_t, Q_t\right)  - \bar h^0 \left(\delta^0_t,  \bar \delta^1_t, Z^0_t, Q_t\right) \Big)\mathrm dt.
    \end{align*}
    But since $\left(U^{0,\ \bar \delta}_t\right)_{t\in [0,T]}$ is a $\mathbb P^{\bar \delta}-$martingale, we must have
    $$\bar h^0 \left(\bar \delta^0_t,  \bar \delta^1_t, Z^0_t, Q_t\right) = \bar H^0 \left( \bar \delta^1_t, Z^0_t, Q_t\right), \qquad \text{i.e.} \qquad \bar \delta^0_t \in \underset{\delta^0 \in \mathcal B_{\infty}^2(Q^0_t)}{\mathrm{argmax}}\ h^0(\delta^0, \bar \delta^1_t, Z^0_t, Q_t).$$

    \item Note that, although we only considered the case of market maker $M_0$ so far, similar results can also be obtained for $M_1$, and in particular
    $$\bar \delta^1_t \in \underset{\delta^1
    \in \mathcal B_{\infty}^2(Q^1_t)}{\mathrm{argmax}}\ h^1(\bar \delta^0, \delta^1_t, Z^1_t, Q_t) \qquad \forall t \in [0,T],$$

    where $Z^1$ is defined in the case of market maker $M_1$ as $Z^0$ for $M_0.$ This means that $\forall t\in [0,T]$, $\bar \delta_t \in \mathcal E (Z_t,Q_t)$, and by \Cref{FPMM} we get
    $$h^0 \left(\bar \delta^0_t,  \bar \delta^1_t, Z^0_t, Q_t\right) =  H^0 \left( \bar \delta^1_t, Z^0_t, Q_t\right) = \mathcal H^0 \left( Z^0_t, Q_t\right),$$
    and similarly
    $$h^1 \left(\bar \delta^0_t,  \bar \delta^1_t, Z^1_t, Q_t\right) =  H^1 \left( \bar \delta^0_t, Z^1_t, Q_t\right) = \mathcal H^1 \left( Z^1_t, Q_t\right).$$

    \item It remains to prove that $Z = (Z^0, Z^1) \in \mathcal Z$. By \Cref{leminteg}, we know $\exists\ q>0$ such that
    $$\underset{\delta^0 \in \mathcal A^0(\bar \delta^1)}{\sup} \mathbb E^{\delta^0 \otimes \bar \delta^1} \left[ \underset{t \in [0,T]}{\sup} \left| U^{0,\delta^0 \otimes \bar \delta^1}_t \right|^{1+q}\right]<+\infty.$$

    Using Hölder inequality, and the fact that
    $$e^{-\gamma^0 Y_{M_0}\left(t, \bar \delta^1, \xi^0\right)} = -U^{0,\delta^0 \otimes \bar \delta^1}_t   e^{\gamma^0 \bigg(\int_0^t \Big( \delta^{a,0}_s - \beta \big( \delta^{a,0}_s - \underline{\delta^{a,0}_s \otimes \bar \delta^{a,1}_s}  \big) \Big) \mathrm dN^{a,0}_s + \int_0^t \Big( \delta^{b,0}_s - \beta \big(  \delta^{b,0}_s - \underline{\delta^{b,0}_s \otimes \bar \delta^{b,1}_s} \big) \Big) \mathrm dN^{b,0}_s + \int_0^t Q^{0}_s \mathrm dS_s\bigg)},$$
    we get the desired result for $Z^0$. The result for $Z^1$ is obtained similarly, and finally  $Z \in \mathcal Z$.\\

    This proves that $\mathcal C \subset \Xi$.

    \item Let us now consider  $\xi=\left(\xi^0, \xi^1\right) = \left(Y_T^{0, Y^0_0, Z},  Y_T^{1, Y^1_0, Z} \right)$ for some $(Y_0, Z) \in \mathbb R^2 \times \mathcal Z$. By what precedes and using \Cref{FPMM}, it is then easy to prove \Cref{NashMM}. In particular, $\overline{NE}(\xi) \neq \emptyset$, and therefore $\mathcal C = \Xi$. 
\end{enumerate}

\bibliographystyle{plain}

\end{document}